\documentclass[a4paper,UKenglish,cleveref,hyperref,autoref]{lipics-v2021}

\usepackage{enumitem}
\usepackage{complexity}

\usepackage[utf8]{inputenc}
\usepackage[T1]{fontenc}
\usepackage{amsmath}
\usepackage{amssymb}
\usepackage{graphicx}
\usepackage{subcaption}
\usepackage{xspace}
\usepackage{verbatim}
\usepackage{gensymb}
\usepackage{tikz}
\usepackage{environ}
\makeatletter
\newsavebox{\measure@tikzpicture}
\NewEnviron{scaletikzpicturetowidth}[1]{%
  \def\tikz@width{#1}%
  \def\tikzscale{1}\begin{lrbox}{\measure@tikzpicture}%
  \BODY
  \end{lrbox}%
  \pgfmathparse{#1/\wd\measure@tikzpicture}%
  \edef\tikzscale{\pgfmathresult}%
  \BODY
}
\makeatother

\usepackage{setspace}
\usetikzlibrary{snakes}
\usepackage[framemethod=TikZ]{mdframed}

\renewcommand{\leq}{\leqslant}

\mdfdefinestyle{MyFrame}{
	linecolor=black,
	innerbottommargin=\baselineskip,
	backgroundcolor=gray!10!white}

\usepackage[textsize=footnotesize,color=green!40]{todonotes}
\usepackage{marvosym}

\newcommand{\PMC}{\textsc{Perfect Matching Cut}\xspace}
\newcommand{\sPMC}{\textsc{PMC}\xspace}

\newcommand{\var}{\text{var}}

\definecolor{dartmouthgreen}{rgb}{0.05, 0.5, 0.06}

\title{Cutting Barnette graphs perfectly is hard}

\titlerunning{Cutting Barnette graphs perfectly is hard}

\author{\'{E}douard Bonnet}{Univ Lyon, CNRS, ENS de Lyon, Université Claude Bernard Lyon 1, LIP UMR5668, France \and \url{http://perso.ens-lyon.fr/edouard.bonnet/}}{edouard.bonnet@ens-lyon.fr}{https://orcid.org/0000-0002-1653-5822}{}
\author{Dibyayan Chakraborty}{Univ Lyon, CNRS, ENS de Lyon, Université Claude Bernard Lyon 1, LIP UMR5668, France}{dibyayan.chakraborty@ens-lyon.fr}{https://orcid.org/0000-0003-0534-6417}{}
\author{Julien Duron}{Univ Lyon, CNRS, ENS de Lyon, Université Claude Bernard Lyon 1, LIP UMR5668, France}{julien.duron@ens-lyon.fr}{}{}

\authorrunning{\'E. Bonnet, D. Chakraborty, J. Duron}

\Copyright{Édouard Bonnet, Dibyayan Chakraborty, Julien Duron}

\ccsdesc[100]{Theory of computation → Computational complexity and cryptography → Problems, reductions and completeness}

\keywords{Perfect matching, cutset, perfect matching set, planar graphs, Barnette graphs, NP-completeness}

\category{}

\relatedversion{}

\supplement{}

\funding{This work was supported by the ANR projects TWIN-WIDTH (ANR-21-CE48-0014) and Digraphs (ANR-19-CE48-0013).}

\acknowledgements{We are much indebted to Carl Feghali for introducing us to the topic of (perfect) matching cuts, and presenting us with open problems that led to the current paper.
We also wish to thank him and Kristóf Huszár for helpful discussions on an early stage of the project.}

\nolinenumbers 

\hideLIPIcs  

\EventEditors{John Q. Open and Joan R. Access}
\EventNoEds{2}
\EventLongTitle{42nd Conference on Very Important Topics (CVIT 2016)}
\EventShortTitle{CVIT 2016}
\EventAcronym{CVIT}
\EventYear{2016}
\EventDate{December 24--27, 2016}
\EventLocation{Little Whinging, United Kingdom}
\EventLogo{}
\SeriesVolume{42}
\ArticleNo{23}

\begin{document}

\maketitle              

\begin{abstract}
  A~\emph{perfect matching cut} is a perfect matching that is also a~cutset, or equivalently a perfect matching containing an even number of edges on every cycle.
  The corresponding algorithmic problem, \textsc{Perfect Matching Cut}, is known to be NP-complete in subcubic bipartite graphs [Le \& Telle, TCS '22] but its complexity was open in planar graphs and in cubic graphs.
  We settle both questions at once by showing that \textsc{Perfect Matching Cut} is NP-complete in 3-connected cubic bipartite planar graphs or \emph{Barnette graphs}.
  Prior to our work, among problems whose input is solely an undirected graph, only \textsc{\mbox{Distance-2} \mbox{4-Coloring}} was known NP-complete in Barnette graphs.
  Notably, \textsc{Hamiltonian Cycle} would only join this private club if Barnette's conjecture were refuted. 
\end{abstract}

\section{Introduction}

Deciding if an input graph admits a~perfect matching, i.e., a~subset of its edges touching each of its vertices exactly once, notoriously is a~tractable task.
There is indeed a~vast literature, starting arguably in 1947 with Tutte's characterization via determinants~\cite{Tutte47}, of polynomial-time algorithms deciding \textsc{Perfect Matching} (or returning actual solutions) and its optimization generalization \textsc{Maximum Matching}.

In this paper, we are interested in another containment of a~spanning set of disjoint edges --perfect matching-- than as a subgraph.
As containing such a~set of edges as an induced subgraph is a trivial property\footnote{Note however that the induced variant of \textsc{Maximum Matching} is an interesting problem that happens to be NP-complete~\cite{Stockmeyer82}.} (only shared by graphs that are themselves disjoint unions of edges), the meaningful other containment is as a~\emph{semi-induced subgraph}.
By that we mean that we look for a~bipartition of the vertex set or \emph{cut} such that the edges of the perfect matching are ``induced'' in the corresponding cutset (i.e., the edges going from one side of the bipartition to the other), while we do not set any requirement on the presence or absence of edges within each side of the bipartition.

This problem was in fact introduced as the \PMC (\sPMC for short) problem\footnote{The authors consider the framework of $(k,\sigma,\rho)$-partition problem, where $k$ is a~positive integer, and $\sigma, \rho$ are sets of non-negative integers, and one looks for a vertex-partition into $k$ parts such that each vertex of each part has a number of neighbors in its own part in $\sigma$, and a number of other neighbors in $\rho$; hence, \sPMC is then the $(2,\mathbb N,\{1\})$-partition problem.} by Heggernes and Telle who show that it is NP-complete~\cite{Heggernes98}.
As the name \PMC suggests, we indeed look for a perfect matching that is also a~cutset.
Le and Telle further show that \sPMC remains NP-complete in subcubic bipartite graphs of arbitrarily large girth, whereas it is polynomial-time solvable in a superclass of chordal graphs, and in graphs without a particular subdivided claw as an induced subgraph~\cite{Le22}.
An in-depth study of the complexity of \sPMC when forbidding a~single induced subgraph or a finite set of subgraphs has been carried out~\cite{Feghali22,Lucke22b}.

We look at Le and Telle's hardness constructions and wonder what other properties could make \sPMC tractable (aside from chordality, and forbidding a~finite list of subgraphs or induced subgraphs).
A~simpler reduction for bipartite graphs is first presented.
Let us briefly sketch their reduction (without thinking about its correctness) from~\textsc{Monotone Not-All-Equal 3-SAT}, where given a~negation-free 3-CNF formula, one seeks a~truth assignment that sets in each clause a variable to true and a variable to false.
Every variable is represented by an edge, and each 3-clause, by a~(3-dimensional) cube with three anchor points at three pairwise non-adjacent vertices of the cube.
One endpoint of the variable gadget is linked to the anchor points corresponding to this variable among the clause gadgets.
Note that this construction creates three vertices of degree~4 in each clause gadget, and vertices of possibly large degree in the variable gadgets.
Le and Telle then reduce the maximum degree to at most~3, by appropriately subdividing the cubes and tweaking the anchor points, and replacing the variable gadgets by cycles.

Notably the edge subdivision of the clause gadgets creates degree 2-vertices, which are not easy to ``pad'' with a~third neighbor (even more so while keeping the construction bipartite). 
And indeed, prior to our work, the complexity of \sPMC in cubic graphs was open.
Let us observe that on cubic graphs, the problem becomes equivalent to partitioning the vertex set into two sets each inducing a~disjoint union of (independent) cycles.
The close relative, \textsc{Matching Cut}, where one looks for a mere matching that is also a cutset, while NP-complete in general~\cite{Chvatal84}, is polynomial-time solvable in \emph{subcubic} graphs~\cite{Moshi89,Bonsma09}.
The complexity of \textsc{Matching Cut} has further been examined in subclasses of planar graphs~\cite{Patrignani01,Bonsma09}, when forbidding some (induced) subgraphs~\cite{Feghali22,Lucke22,Lucke22b,Feghali23}, on graphs of bounded diameter~\cite{Lucke22,Le19}, and on graphs of large minimum degree~\cite{Chen21}.
\textsc{Matching Cut} has also been investigated with respect to parameterized complexity, exact exponential algorithms~\cite{Kratsch16,Komusiewicz20}, and enumeration~\cite{Golovach22}.

It was also open if \sPMC is tractable on planar graphs.
Note that Bouquet and Picouleau show that~a related problem, \textsc{Disconnected Perfect Matching}, where one looks for a perfect matching that contains a~cutset, is NP-complete on planar graphs of maximum degree~4, on planar graphs of girth~5, and on \mbox{5-regular} bipartite graphs~\cite{Bouquet20}.
They incidentally call this related problem~\PMC but subsequent references~\cite{Feghali22,Le22} use the name \textsc{Disconnected Perfect Matching} to avoid confusion.
We will observe that \sPMC is equivalent to asking for a~perfect matching containing an even number of edges from every cycle of the input graph.
The sum of even numbers being even, it is in fact sufficient that the perfect matching contains an even number of edges from every element of a~cycle basis.
There is a~canonical cycle basis for planar graphs: the bounded faces.
This gives rise to the following neat reformulation of \sPMC in planar graphs: is there a perfect matching containing an even number of edges along each face?

While \textsc{Matching Cut} is known to be NP-complete on planar graphs~\cite{Patrignani01,Bonsma09}, it could have gone differently for \sPMC for the following ``reasons.''
\textsc{Not-All-Equal 3-SAT}, which appears as the \emph{right} starting point to reduce to \sPMC, is tractable on planar instances~\cite{Moret88}.
In planar graphs, perfect matchings are \emph{simpler} than arbitrary matchings in that they alone~\cite{Vadhan01} can be counted efficiently~\cite{Temperley61,Kasteleyn67}.
Let us finally observe that \textsc{Maximum Cut} can be solved in polynomial time in planar graphs~\cite{Hadlock75}.

\medskip

In fact, we show that the reformulations for cubic and planar graphs cannot help algorithmically, by simultaneously settling the complexity of \sPMC in cubic and in planar graphs, with the following stronger statement. 
\begin{theorem}\label{thm:hard}
\PMC is NP-hard in 3-connected cubic bipartite planar graphs.
\end{theorem}

Not very many problems are known to be NP-complete in cubic bipartite planar graphs. 
Of the seven problems defined on mere undirected graphs from Karp's list of 21 NP-complete problems~\cite{Karp72}, only \textsc{Hamiltonian Path} is known to remain NP-complete in this class, while the other six problems admit a~polynomial-time algorithm.
Restricting ourselves to problems where the input is purely an~undirected graph,\footnote{Among problems with edge orientations, vertex or edge weights, or prescribed subsets of vertices or edges, the list is significantly longer, and also includes \textsc{Minimum Weighted Edge Coloring}~\cite{Werra09}, \textsc{List Edge Coloring} and \textsc{Precoloring Extension}~\cite{Marx05}, \textsc{$k$-In-A-Tree}~\cite{Derhy09}, etc.} besides \textsc{Hamiltonian Path/Cycle} \cite{Munaro17,Akiyama80}, \textsc{Minimum Independent Dominating Set} was also shown NP-complete in cubic bipartite planar graphs~\cite{Loverov20}, as well as \textsc{$P_3$-Packing} \cite{Kosowski05} (hence, an equivalent problem phrased in terms of disjoint dominating and 2-dominating sets~\cite{Miotk20}), and \textsc{Distance-2 4-Coloring}~\cite{Feder21}.
To our knowledge, \textsc{Minimum Dominating Set} is only known NP-complete in \emph{subcubic} bipartite planar graphs~\cite{Garey79,Korobitsin92}. 

It is interesting to note that the reductions for \textsc{Hamiltonian Path}, \textsc{Hamiltonian Cycle}, \textsc{Minimum Independent Dominating Set}, and~\textsc{$P_3$-Packing} all produce cubic bipartite planar graphs that are \emph{not} 3-connected.
Notoriously, lifting the NP-hardness of \textsc{Hamiltonian Cycle} to the 3-connected case would require to disprove Barnette's conjecture\footnote{which precisely states that every polyhedral (that is, 3-connected planar) cubic bipartite graphs admits a hamiltonian cycle.} (and that would be indeed sufficient~\cite{Feder06}).
Note that hamiltonicity in cubic graphs is equivalent to the existence of a~perfect matching that is \emph{not} an edge cut (i.e., whose removal is not disconnecting the graph).
We wonder whether there is something inherently simpler about \mbox{\emph{3-connected}} cubic bipartite planar graphs, which would go beyond hamiltonicity (assuming that Barnette's conjecture is true).

Let us call \emph{Barnette} a~3-connected cubic bipartite planar graph.
It appears that, prior to our work, \textsc{Distance-2 4-Coloring} was the only \emph{vanilla} graph problem shown NP-complete in Barnette graphs~\cite{Feder21}.
Arguing that \textsc{\mbox{Distance-2} 4-Coloring} is a~problem on \emph{squares} of Barnette graphs more than it is on Barnette graphs, a~case can be made for \PMC to be the first natural problem proven NP-complete in Barnette graphs.


\medskip

\textbf{Provably tight subexponential-time algorithm.}
Note that our reduction together with existing results and a~known methodology give a~fine-grained understanding, under the Exponential-Time Hypothesis\footnote{the assumption that there is a~$\lambda > 0$ such that no algorithm solves $n$-variable \textsc{3-SAT} in time $\lambda^n n^{O(1)}$} (or ETH)~\cite{Impagliazzo01}, on solving \PMC in planar graphs.

On the algorithmic side, there is a~$2^{O(\sqrt n)}$-time algorithm for \sPMC in \mbox{$n$-vertex} planar graphs, as a~consequence of a~$2^{O(w)}n^{O(1)}$-time algorithm for \mbox{$n$-vertex} graphs given with a~tree-decomposition of width~$w$, and the fact that tree-decompositions of width $O(\sqrt n)$ always exist in planar graphs and can be computed in polynomial-time~\cite{Lipton80}.
The $2^{O(w)}n^{O(1)}$-time algorithm can be obtained directly or as a~consequence of a~result of Pilipczuk~\cite{Pilipczuk11} that any problem expressible in Existential Counting Modal Logic (ECML) admits a single-exponential fixed-parameter algorithm in treewidth.
ECML allows existential quantifications over vertex and edge sets followed by a~counting modal formula to be satisfied \emph{from every vertex}.
Counting modal formulas enrich quantifier-free Boolean formulas with $\Diamond^S \varphi$, whose semantics is that the current vertex~$v$ has a number of neighbors satisfying $\varphi$ in the ultimately periodic set $S$ of non-negative integers. 
One can thus express \PMC in ECML as $$\exists X \subseteq V(G), \forall v \in V(G),~G,X,v \models X \rightarrow \Diamond^{\{1\}}(\neg X)~\land~\neg X \rightarrow \Diamond^{\{1\}}X,$$
which states that there is a set $X$ such that every vertex in $X$ has exactly one neighbor outside $X$, and vice versa.

On the complexity side, the Sparsification lemma~\cite{sparsification}, the folklore linear reductions from bounded-occurrence \textsc{3-SAT} to bounded-occurrence \textsc{Monotone Not-All-Equal 3-SAT} and to \textsc{Monotone Not-All-Equal 3-SAT-E4}~\cite{darmann2020simple}, and finally our quadratic reduction, imply that $2^{\Omega(\sqrt n)}$ time is required to solve \sPMC in $n$-vertex planar graphs.
Our reduction (as we will see) indeed has a~quadratic blow-up as it creates $O(1)$ vertices per variable and clause, and $O(1)$ vertices for each of the $O(n^2)$ crossings in a (non-planar) drawing of the variable-clause incidence graph. 

\medskip

\textbf{Outline of the proof.}
We reduce the NP-complete problem \textsc{Monotone Not-All-Equal 3-SAT} with exactly 4~occurrences of each variable~\cite{darmann2020simple} to~\sPMC.
Observe that flipping the value of every variable of a~satisfying assignment results in another satisfying assignment.
We thus see a~solution to \textsc{Monotone Not-All-Equal 3-SAT} simply as a bipartition of the set of variables.

As we already mentioned, \textsc{Not-All-Equal 3-SAT} restricted to planar instances (i.e., where the variable-clause incidence graph is planar) is in P.
We thus have to design \emph{crossing} gadgets in addition to \emph{variable} and \emph{clause} gadgets.
Naturally our gadgets are bipartite graphs with vertices of degree 3, except for some special \emph{anchors}, vertices of degree 2 with one incident edge leaving the gadget.

The variable gadget is designed so that there is a unique way a~perfect matching cut can intersect it.
It might seem odd that no ``binary choice'' happens within it.
The role of this gadget is only to serve as a baseline for which side of the bipartition the variable lands in, while the ``truth assignments'' take place in the clause gadgets. 
(Actually the same happens with Le and Telle's first reduction~\cite{Le22}, where the variable gadget is a single edge, which has to be in any solution.)

Our variable gadget consists of 36 vertices, including 8 anchor points; see~\Cref{fig:var-gadget}.
(We will later explain why we have 8 anchor points and not simply 4, that is, one for each occurrence of the variable.)
Note that in all the figures, we adopt the following convention:
\begin{itemize}
\item black edges cannot (or can no longer) be part of a~perfect matching cut,
\item red edges are in every perfect matching cut, 
\item each blue edge $e$ is such that at least one perfect matching cut within its gadget includes~$e$, and at least one excludes $e$, and
\item brown edges are blue edges that were indeed chosen in the solution.
\end{itemize}
Let us recall that \sPMC consists of finding a perfect matching containing an even number of edges from each cycle.
Thus we look for a~perfect matching $M$ such that every path (or walk) between $v$ and $w$ contains a~number of edges of $M$ whose parity only depends on $v$ and $w$.
If this parity is even $v$ and $w$ are on the \emph{same side}, and if it is odd, $v$ and $w$ are on \emph{opposite sides}.
The 8 anchor points of each variable gadget are forced on the same side.
This is the \emph{side of the variable}.

At the core of the clause gadget is a subdivided cube of blue edges; see \Cref{fig:clause-gadget}.
There are three vertices ($u_1, u_8, u_{14}$ on the picture) of the subdivided cube that are forced on the same side as the corresponding three variables.
Three perfect matching cuts are available in the clause gadget, each separating (i.e., putting on opposite sides) a~different vertex of $\{u_1, u_8, u_{14}\}$ from the other two.
Note that this is exactly the semantics of a~not-all-equal 3-clause.
We in fact need two copies of the subdivided cube, partly to increase the degree of some subdivided vertices, partly for the same reason we duplicated the anchor vertices in the variable gadgets.
(The latter will be explained when we present the crossing gadgets.)
Increasing the degree of all the subdivided vertices complicate further the gadget and create two odd faces.
Fortunately these two odd faces have a~common neighboring even face.
We can thus ``fix'' the parity of the two odd faces by plugging the sub-gadget $D_j$ in the even face.
We eventually need a total of 112 vertices, including 6 anchor points.

Let us now describe the crossing gadgets.
Basically we want to replace every intersection point of two edges by a~4-vertex cycle.
This indeed propagates black edges (those that cannot be in any solution).
The issue is that going through such a~crossing gadget flips one's side.
As we cannot guarantee that a~variable ``wire'' has the same parity of intersection points towards each clause gadget it is linked to, we duplicate these wires.
At a~previous intersection point, we now have two parallel wires crossing two other parallel wires, making four crossings.
The gadget simply consists of four 4-vertex cycles; see~\Cref{fig:cross-replace}.
Check in~\Cref{fig:cross-type} that the sides are indeed preserved.
This explains why we have 8 anchor points (not~4) in each variable gadget, and 6 anchor points (not~3) in each clause gadget. 

\section{Preliminaries}

For a graph $G$, we denote by $V(G)$ its set of vertices and by $E(G)$ its set of edges.
If $U \subseteq V(G)$, the \emph{subgraph of $G$ induced by $U$}, denoted $G[U]$ is the graph obtained from $G$ by removing the vertices not in $U$.
$E_G(U)$ (or $E(U)$ when $G$ is clear) is a shorthand for $E(G[U])$.
For $M \subset E(G)$, $G - M$ is the spanning subgraph of $G$ obtained by removing the edges in $M$ (while preserving their endpoints).
A connected component of $G$ is a maximal set $U \subseteq V(G)$ such that $G[U]$ is connected.
A graph $G$ is \emph{cubic} if every vertex of $G$ has exactly three neighbors.
A graph is \emph{bipartite} if it contains no odd cycles.
We may use \emph{$k$-cycle} as a short-hand for the $k$-vertex cycle.

Given two disjoint sets $X, Y \subseteq V(G)$ we denote by $E(X, Y)$ the set of edges between $X$ and $Y$.
A set $M \subseteq E(G)$ is a~\emph{cutset}\footnote{We avoid using the term ``edge cut'' since, for some authors, an edge cut is, more generally, a~subset of edges whose deletion increases the number of connected components.} of $G$ if there is a~proper bipartition $X \uplus Y = V(G)$, called~\emph{cut}, such that $M = E(X, Y)$.
Note that a~cut fully determines a~cutset, and among connected graphs a cutset fully determines a~cut.
When dealing with connected graphs, we may speak of \emph{the} cut of a cutset.
For $X \subseteq V(G)$ the set of \emph{outgoing edges of $X$} is $E(X, V(G) \setminus X)$.
For a cutset $M$ of a~connected graph $G$, and $u, v \in V(G)$, we say that $u$ and $v$ are on the \emph{same side} (resp. on \emph{opposite sides}) of $M$ if $u$ and $v$ are on the same part (resp. on different parts) of the cut of $M$.

A~\emph{matching} (resp. \emph{perfect matching}) of $G$ is a set $M \subset E(G)$ such that each vertex of $G$ is incident to at most (resp. exactly) one edge of $M$.
A~\emph{perfect matching cut} is a perfect matching that is also a~cutset.
For $M \subseteq E(G)$ and $U \subseteq V(G)$, we say that $M$ is a~\emph{perfect matching cut of $G[U]$} if $M \cap E(U)$ is so.



A graph is \emph{planar} if it can be embedded in the plane, i.e., drawn such that edges (simple curves) may only intersect at their endpoint (the vertices).
A \emph{plane graph} is a planar graph together with such an embedding.
Given a plane graph $G$, a face of $G$ is a connected component of the plane after removing the embedding of~$G$.
A~\emph{facial cycle} of a~plane graph $G$ is a cycle of $G$ that bounds a~face of~$G$.
We say that two plane graphs $G$ and $H$ are \emph{translates} if the embedding of $G$ is a~translate of the embedding of $H$.

\section{Proof of \Cref{thm:hard}}

\newcommand{\MSAT}{\textsc{Monotone Not-All-Equal 3SAT-E4}\xspace}
\newcommand{\cycle}[2]{S_{#1}^{#2}}
\newcommand{\varGadget}[1]{\mathcal{X}_{#1}}
\newcommand{\clauseGadget}[1]{\mathcal{C}_{#1}}
\newcommand{\construct}{G(I)}
\newcommand{\cubic}{H(I)}

\newcommand{\topvarclause}[2]{t_{#1,#2}}
\newcommand{\topclausevar}[2]{t'_{#1,#2}}
\newcommand{\botvarclause}[2]{b_{#1,#2}}
\newcommand{\botclausevar}[2]{b'_{#1,#2}}

\newcommand{\card}[1]{\left|#1\right|}
\newcommand{\fourcycle}[2]{F\left(#1,#2\right)}
\newcommand{\fin}[1]{f^{in}_{#1}}
\newcommand{\fout}[1]{f^{out}_{#1}}

\newcommand{\internalclausEdge}[2]{e_{#1}^{#2}}

Before we give our reduction, we start with a~handful of useful lemmas and observations, which we will later need.

\subsection{Preparatory lemmas}

\begin{lemma}\label{lem:pmc-cycles}
  Let $G$ be a~graph, and $M \subseteq E(G)$.
  Then $M$ is a~cutset if and only if for every cycle $C$ of $G$, $\card{E(C) \cap M}$ is even.
\end{lemma}

\begin{proof}
  Suppose that $M$ is a cutset, and let $(A,B)$ be a~cut of $M$.
  Every closed walk (and in particular, cycle) contains an even number of edges of $M$, since the edges of $M$ go (along the walk) from $A$ to $B$, and from $B$ to $A$.

  Now assume that every cycle of $G$ has an even number of edges in common with $M$.
  We build a~cut $(A,B)$.
  For each connected component $H$ of $G$, we fix an arbitrary vertex $v \in V(H)$, and do the following.
  For each vertex $w \in V(H)$, put $w$ in $A$ if there is a~path from $v$ to $w$ taking an even number of edges from $M$, and in $B$ if there is a path from $v$ to $w$ taking an odd number of edges from $M$.
  It holds that $A \cup B=V(G)$.
  By our assumption on the cycles of $G$, $A \cap B = \emptyset$.
  Hence $(A,B)$ is indeed a~cut.
  The cutset of $(A,B)$ is, by construction,~$M$.
\end{proof}

\begin{lemma}\label{lem:facial-cycle}
  Let $G$ be a plane graph, and $M \subseteq E(G)$.
  Then $M$ is a cutset if and only if for any facial cycle $C$ of $G$, $\card{E(C) \cap M}$ is even.
\end{lemma}
\begin{proof}
  The forward implication is a direct consequence of~\Cref{lem:pmc-cycles}.
  The converse comes from the known fact that the bounded faces form a~cycle basis; see for instance~\cite{Diestel12}.

  If $H$ is a subgraph of $G$, let $\tilde{H}$ be the vector of $\mathbb F_2^{E(G)}$ with 1 entries at the positions corresponding to edges of $H$.
  Thus, for any cycle $C$ of $G$, we have $\tilde{C} = \Sigma_{1 \leqslant i \leqslant k} \tilde{F_i}$ where $F_i$ are facial cycles of $G$.
  And $|M \cap E(C)|$ has the same parity as $\Sigma_{1 \leqslant i \leqslant k} |M \cap E(F_i)|$, a sum of even numbers.
\end{proof}

\begin{lemma} \label{lem:4-cycle}
  Let $M$ be a perfect matching cut of a cubic graph $G$.
  Let $C$ be an induced $4$-vertex cycle of~$G$.
  Then, exactly one of the following holds:
\begin{enumerate}[label=(\alph*)]
    \item\label{it:4-cycle-1} $E(C)\cap M = \emptyset$ and the four outgoing edges of $V(C)$ belong to~$M$. 
    \item\label{it:4-cycle-2} $|E(C) \cap M| = 2$, the two edges of $E(C) \cap M$ are disjoint, and none of the outgoing edges of $V(C)$ belongs to $M$.
\end{enumerate}
\end{lemma}
\begin{proof}
The number of edges of $M$ within $E(C)$ is even by~\Cref{lem:facial-cycle}.
Thus $|E(C) \cap M| \in \{0, 2\}$, as all four edges of $E(C)$ do not make a~matching.

Suppose that $E(C) \cap M = \emptyset$.
As $M$ is a perfect matching, for every $v \in V(C)$ there is an edge in $M$ incident to $v$ and not in $E(C)$.
As $G$ is cubic, every outgoing edge of $V(C)$ is in $M$.

Suppose instead that $|E(C) \cap M| = 2$.
As $M$ is a~matching, the two edges of $E(C) \cap M$ do not share an endpoint.
It implies that all the four vertices of $C$ are touched by these two edges.
Thus no outgoing edge of $V(C)$ can be in $M$.
\end{proof}

\begin{corollary}\label{cor:4-cycle-propagation}
  Let $M$ be a perfect matching of a cubic graph $G$.
  Let $C_1$, $C_2$ two vertex-disjoint induced $4$-vertex cycles of $G$ such that there is an edge between $V(C_1)$ and $V(C_2)$.
  Then $E(C_1) \cap M \neq \emptyset$ if and only if $E(C_2) \cap M \neq \emptyset$.
\end{corollary}
\begin{proof}
  Suppose $E(C_1) \cap M \neq \emptyset$.
  By \Cref{lem:4-cycle} on $C_1$, no outgoing edge of $V(C_1)$ is in $M$.
  Thus, there is an outgoing edge of $V(C_2)$ that is not in $M$.
  Applying \Cref{lem:4-cycle} on $C_2$, we have $E(C_2) \cap M \neq \emptyset$.
  We get the converse symmetrically.
\end{proof}

\begin{lemma}\label{lem:hex-3-edges}
Let $M$ be a~perfect matching cut of a~cubic graph~$G$.
If a~$6$-cycle has three outgoing edges in $M$, then all six outgoing edges are in $M$.
\end{lemma}
\begin{proof}
\sloppy Let $C$ be our $6$-cycle.
Remember that, as $M$ is a perfect matching cut, $|E(C) \cap M|$ is even.
This means that $|E(C) \cap M|$ is either $0$ or $2$.
If $|E(C) \cap M| = 2$, four vertices of $C$ are touched by $E(C) \cap M$, which rules out that three outgoing edges of $V(C)$ are in $M$.
Thus $E(C) \cap M = \emptyset$ and, $G$ being cubic, every outgoing edge of $V(C)$ is in $M$. 
\end{proof}

\begin{lemma}\label{lem:hex-incident-square}
Let $M$ a perfect matching cut of~a cubic bipartite graph~$G$.
Suppose $C$ is a $6$-cycle $v_1v_2 \ldots v_6$ of $G$, such that $v_2v_3$, $v_3v_4$, $v_5v_6$ and $v_6v_1$ are in some induced 4-cycles.
Then $M \cap E(C) = \emptyset$.
\end{lemma}
\begin{proof}
By applying~\Cref{lem:4-cycle} on the 4-cycle containing $v_2v_3$, and the one containing $v_6v_1$, it holds that $v_1v_2 \in M \Leftrightarrow v_3v_4 \in M \Leftrightarrow v_5v_6 \in M$.
Thus none of these three edges can be in $M$, because $C$ would have an odd number of edges in $M$.
Symmetrically, no edge among  $v_2v_3$, $v_4v_5$ and $v_6v_1$ can be in $M$.
Thus no edge of $C$ is in $M$.
\end{proof}

\begin{observation}\label{obs:sides}
Let $G$ be a graph and $M$ be a perfect matching cut of $G$. Let $u, v$ be two vertices of $G$. Then for any path $P$ between $u$ and $v$, 
$|E(P) \cap M|$ is even if and only if $u$ and $v$ are on the same side of $M$. Note that implies that for any paths $P, Q$ from $u$ to $v$, $|E(P) \cap M|$ and $|E(Q) \cap M|$ have same parity.
\end{observation}

\subsection{Reduction}

We will prove \Cref{thm:hard} by reduction from the NP-complete \MSAT~\cite{darmann2020simple}.
In \MSAT, the input is a~3-CNF formula where each variable occurs exactly four times, each clause contains exactly three distinct literals, and no clause contains a negated literal.
Here we say that a~truth assignment on the variables \emph{satisfies} a clause $C$ if at least one literal of $C$ is true and at least least one literal of $C$ is false.
The objective is to decide whether there is a truth assignment that satisfies all clauses.
We can safely assume (and we will) that the variable-clause incidence graph $\text{inc}(I)$ of $I$ has no cutvertex among its ``variable'' vertices.
Indeed the reduction from \textsc{Monotone Not-All-Equal 3-SAT} to its four-occurrence variant does not create such cutvertices if they do not exist originally.
Now if there is a~``variable'' cutvertex $v$ in a~\textsc{Monotone Not-All-Equal 3-SAT}-instance $J$, one can split $J$ into 
$J_1$ made of one connected component $X$ of $\text{inc}(J)-\{v\}$ plus $v$, and $J_2$ made of $\text{inc}(J) \setminus X$.
One can observe that $J$ is positive if and only if $J_1$ and $J_2$ are positive.
As $\text{inc}(J_1)$ and $\text{inc}(J_2)$ sum up to one more vertex than $\text{inc}(J)$, such a~scheme is a polynomial-time Turing reduction to subinstances without ``variable'' cutvertices.

Let $I$ be an instance of \MSAT~with variables $x_1,x_2,\ldots,$ $x_n$ and clauses $m=4n/3$ clauses $C_1,C_2,\ldots,C_m$.
We shall construct, in polynomial time, an equivalent \sPMC-instance~$\construct$ that is~Barnette.

Our reduction consists of three steps.
First we construct a cubic graph $\cubic$ by introducing \emph{variable gadgets} and \emph{clause gadgets}. 
Then we \emph{draw} $\cubic$ on the plane, i.e., we map the vertices of $\cubic$ to a set of points on the plane, and the edges of $\cubic$ to a set of simple curves on the plane. 
We shall refer to this drawing as $\mathcal{R}$.
Note that, this drawing may not be planar, {i.e.,} two simple curves (or analogously the corresponding edges) might intersect at a point which is not their endpoints. 
Finally, we eliminate the crossing points by introducing \emph{crossing gadgets}.
(Recall that if the variable-clause incidence graph of a~\textsc{Not-All-Equal 3-SAT} instance is planar, then its satisfiability can be tested in polynomial time~\cite{Moret88}; hence, we do need crossing gadgets.)  
The resulting graph $\construct$ is Barnette, and we shall prove that $\construct$ has a perfect matching if and only if $I$ is a positive instance of \MSAT.
We now describe the above steps.

\begin{figure}[h!]
    \centering
    \begin{scaletikzpicturetowidth}{\textwidth}
    \begin{tikzpicture}[scale=\tikzscale]
     
    \foreach \x/\y/\w/\z [count = \n] in
 {0/0/0.5/0.5, 0.5/0.5/1/0.5, 1/0.5/1.5/0, 1.5/0/1/-0.5, 1/-0.5/0.5/-0.5, 0/0/0.5/-0.5,
  3/0/3.5/0.5, 3.5/0.5/4/0.5, 4/0.5/4.5/0, 4.5/0/4/-0.5, 4/-0.5/3.5/-0.5, 3/0/3.5/-0.5,
  6/0/6.5/0.5, 6.5/0.5/7/0.5, 7/0.5/7.5/0, 7.5/0/7/-0.5, 7/-0.5/6.5/-0.5, 6/0/6.5/-0.5, 
    9/0/9.5/0.5, 9.5/0.5/10/0, 10/0/9.5/-0.5, 9.5/-0.5/9/0,
    11/0/11.5/0.5, 11.5/0.5/12/0.5, 12/0.5/12.5/0, 12.5/0/12/-0.5, 12/-0.5/11.5/-0.5, 11.5/-0.5/11/0, -1/-2/-0.5/-2, 5/-2/5.5/-2,
    8/-2/8.5/-2, 13/-2/13.5/-2
 }
    {
    	\draw (\x, \y) -- (\w,\z);
    }
    
        \foreach \x/\y/\w/\z [count = \n] in
 {0/0/-0.5/0, -0.5/0/-0.5/-2, 
 0.5/0.5/-1/0.5, -1/0.5/-1/-2, 
 0.5/-0.5/0.5/-1.5, 0.5/-1.5/5/-1.5, 5/-1.5/5/-2,
 1/0.5/3.5/0.5, 1.5/0/3/0, 1/-0.5/3.5/-0.5,
 4/0.5/6.5/0.5, 4/-0.5/6.5/-0.5, 4.5/0/6/0,
 7/-0.5/7/-1.5, 7/-1.5/5.5/-1.5, 5.5/-1.5/5.5/-2,
 7/0.5/8/0.5, 8/0.5/8/0, 8/0/9/0,
 7.5/0/7.5/-1.5, 7.5/-1.5/8/-1.5, 8/-1.5/8/-2,
 11/0/10/0, 11.5/0.5/9.5/0.5, 11.5/-0.5/9.5/-0.5,
 12/-0.5/12/-1.5, 8.5/-1.5/12/-1.5, 8.5/-2/8.5/-1.5,
 12/0.5/13.5/0.5/, 13.5/-2/13.5/0.5,
 12.5/0/13/0, 13/0/13/-2
 }
    {
    	\draw[very thick, red] (\x, \y) -- (\w,\z);
    }
    
    \foreach \x/\y [count = \n] in
 { 0.75/0, 3.75/0, 6.75/0, 9.5/0, 11.75/0
 }
    {
    	\node at (\x, \y) { $\cycle{i}{\n}$};
    }
    
    \node[right] at (13.5,-2) {$\topvarclause{i}{q}$};
    \node[left] at (13,-2) {$\botvarclause{i}{q}$};
    
    \node[right] at (8.5,-2) {$\topvarclause{i}{p}$};
    \node[left] at (8,-2) {$\botvarclause{i}{p}$};
    
    \node[right] at (5.5,-2) {$\topvarclause{i}{k}$};
    \node[left] at (5,-2) {$\botvarclause{i}{k}$};
    
    \node[right] at (-0.5,-2) {$\topvarclause{i}{j}$};
    \node[left] at (-1,-2) {$\botvarclause{i}{j}$};

    \foreach \i in {-1,-0.5,5,5.5,8,8.5,13,13.5}{
      \draw (\i,-2)--++(0,-0.3) ;
    }
    
    \foreach \x/\y [count = \n] in
 {0/0, 0.5/0.5, 1/0.5, 1.5/0, 1/-0.5, 0.5/-0.5,
 3/0, 3.5/0.5, 4/0.5, 4.5/0, 4/-0.5, 3.5/-0.5,
 6/0, 6.5/0.5, 7/0.5, 7.5/0, 7/-0.5, 6.5/-0.5, 
 9/0, 9.5/0.5, 10/0, 9.5/-0.5,
 11/0, 11.5/0.5, 12/0.5, 12.5/0, 12/-0.5, 11.5/-0.5, 
 -1/-2, -0.5/-2, 
 5/-2, 5.5/-2,
 8/-2, 8.5/-2,
 13/-2, 13.5/-2
 }
    {
    	\filldraw (\x, \y) circle (3pt);
    }   
    \end{tikzpicture}
    \end{scaletikzpicturetowidth}
    \caption{Variable Gadget $\varGadget{i}$ corresponding to the variable $x_i$ appearing in the clauses $C_j,C_k,C_p,C_q$ with $j<k<p<q$.}
    \label{fig:var-gadget}
\end{figure}
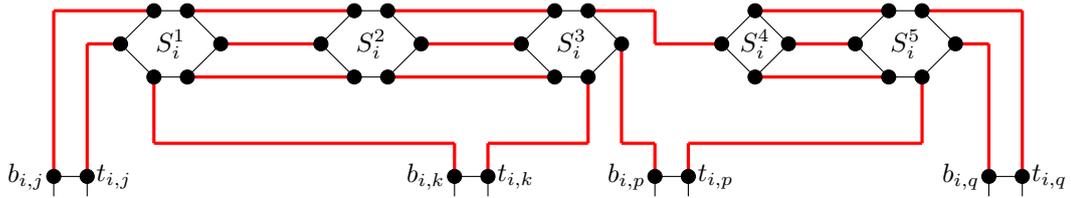

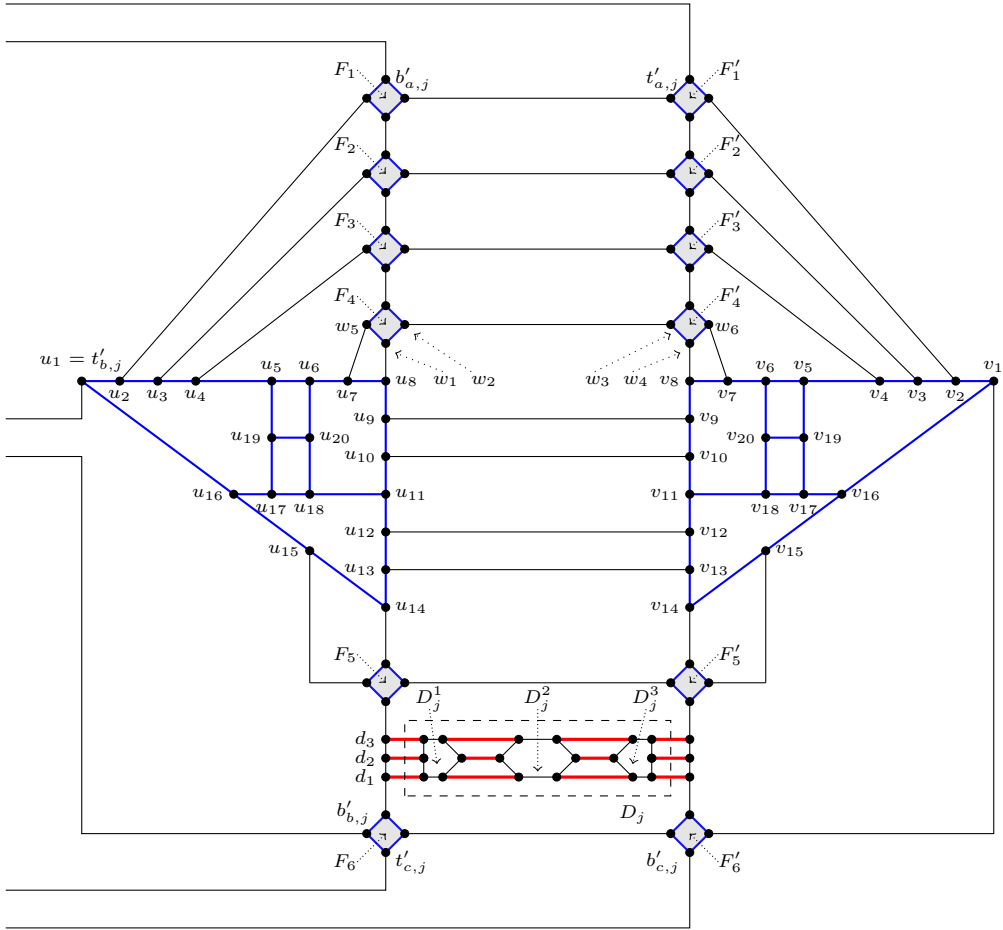
\begin{figure}[h!]
    \centering
        \begin{tikzpicture}[scale = 0.5]

    \foreach \x/\y/\w/\z [count = \n] in
 {0/6/8/6, 8/6/8/0, 8/0/0/6, 4/3/8/3, 5/6/5/3, 6/6/6/3, 5/4.5/6/4.5}
    {
    	\draw[thick,blue] (\x,\y) -- (\w,\z); 
    	\draw[thick,blue] (24 - \x, \y) -- (24 - \w, \z); 
     
    }
    
    \foreach \x/\y [count = \n] in
 {8/5, 8/4, 8/2, 8/1, 8.5/-2, 8.5/-6, 11.5/-3.5, 11.5/-4.5, 8.5/7.5, 8.5/9.5, 8.5/11.5, 8.5/13.5}
    {
    	\draw (\x,\y) -- (24-\x,\y); 
     
    }

    \foreach \x/\y/\w/\z [count = \n] in
 {8/13/8.5/13.5, 8.5/13.5/8/14, 8/14/7.5/13.5,  7.5/13.5/8/13,
  8/11/8.5/11.5, 8.5/11.5/8/12, 8/12/7.5/11.5, 7.5/11.5/8/11,
  8/9/8.5/9.5, 8.5/9.5/8/10, 8/10/7.5/9.5, 7.5/9.5/8/9,
  8/7/8.5/7.5, 8.5/7.5/8/8, 8/8/7.5/7.5, 7.5/7.5/8/7,
   8/-2.5/8.5/-2, 8.5/ -2/8/-1.5, 8/-1.5/7.5/-2, 7.5/-2/8/-2.5, 8/-6.5/8.5/-6, 8.5/-6/8/-5.5, 8/-5.5/7.5/-6, 7.5/-6/8/-6.5}
    {
    	\draw[thick,blue] (\x,\y) -- (\w,\z); 
        \draw[thick,blue] (24-\x,\y) -- (24-\w,\z);
    }

    \foreach \x/\y/\w/\z/\a/\b/\c/\d [count = \n] in
 {8/13/8.5/13.5/8/14/7.5/13.5,
 8/11/8.5/11.5/8/12/7.5/11.5,
 8/9/8.5/9.5/8/10/7.5/9.5, 
 8/7/8.5/7.5/8/8/7.5/7.5,
 8/-2.5/8.5/-2/8/-1.5/7.5/-2,
 8/-6.5/8.5/-6/8/-5.5/7.5/-6
}
    {
    	\filldraw[gray,opacity=0.2] (\x,\y) -- (\w,\z) -- (\a,\b) -- (\c,\d) -- cycle; 
    	\filldraw[gray,opacity=0.2] (24-\x,\y) -- (24-\w,\z) -- (24-\a,\b) -- (24-\c,\d) -- cycle; 
    }

    \foreach \x/\y/\w/\z [count = \n] in
 { 8/13/8/12, 8/11/8/10, 8/9/8/8, 8/7/8/6, 
 8/-1.5/8/0, 8/-2.5/8/-5.5, 7.5/-2/6/-2, 
 6/-2/6/1.5,  8/-4/9/-4, 9.5/-4.5/10/ -4, 
 10/ -4/9.5/-3.5, 9.5/-3.5/9/-3.5, 9/-3.5/9/-4, 9/-4/9/-4.5, 9/-4.5/10/-4.5, 
 11/-4/11.5/-4.5, 11/-4/11.5/-3.5, 
 7.5/7.5/7/6, 7.5/9.5/3/6, 7.5/11.5/2/6, 7.5/13.5/1/6 }
    {
    	\draw (\x,\y) -- (\w,\z); 
        \draw (24-\x,\y) -- (24-\w,\z);
    }    
    
    \foreach \x/\y/\w/\z [count = \n] in
 { 8/-4/9/-4, 10/-4/11/-4, 9.5/-4.5/11.5/-4.5, 9.5/-3.5/11.5/-3.5,
 9/-4.5/8/-4.5, 9/-3.5/8/-3.5}
    {
    	\draw[very thick, red] (\x,\y) -- (\w,\z); 
        \draw[very thick, red] (24-\x,\y) -- (24-\w,\z);
    } 
    
    
    \draw (24,6) -- (24,-6) -- (16.5,-6);
    \draw (0,6) -- (0,5) -- (-2,5); 
    \draw (7.5,-6) -- (0,-6) -- (0,4) -- (-2,4);
    
    \draw (16,-6.5) -- (16,-8.5) -- (-2,-8.5);
    \draw (8,-6.5) -- (8,-7.5) -- (-2,-7.5);
    
     \draw (16,14) -- (16,16) -- (-2,16);
    \draw (8,14) -- (8,15) -- (-2,15);
    
    
    
    \foreach \x/\y [count = \n] in
  {0/6, 1/6, 2/6, 3/6, 5/ 6, 6/ 6, 7/6, 8/6, 8/5, 8/4, 8/3, 8/2, 8/1, 8/0, 6/1.5, 4/ 3, 5/ 3, 6/ 3, 5/ 4.5, 6/ 4.5, 8/14, 8/-6.75 }
	{
		\ifthenelse{\n=1} {\node[above] at (\x, \y) {\scriptsize $u_{\n} = \scriptsize \topclausevar{b}{j}$}; \node[above] at (24 - \x, \y) {\scriptsize $v_{\n}$}; \node[above] at (7.15, -6) {\scriptsize $\botclausevar{b}{j}$}; }{}
		\ifthenelse{\n>1 \and \n<5 } {\node[below] at (\x, \y) {\scriptsize $u_{\n}$}; \node[below] at (24 - \x, \y) {\scriptsize $v_{\n}$};}{}
		\ifthenelse{\n=17 \OR \n=18 } {\node[below] at (\x, \y) {\scriptsize $u_{\n}$}; \node[below] at (24 - \x, \y) {\scriptsize $v_{\n}$};}{}
		
		\ifthenelse{\n>4 \and \n<7 } {\node[above] at (\x, \y) {\scriptsize $u_{\n}$}; \node[above] at (24 - \x, \y) {\scriptsize $v_{\n}$};}{}
		
		\ifthenelse{\n=7} {\node[below] at (\x, \y) {\scriptsize $u_{\n}$}; \node[below] at (24 - \x, \y) {\scriptsize $v_{\n}$};}{}
		
		\ifthenelse{\n=8 \OR \n=11 \OR \n=14 \OR \n=20 } {\node[right] at (\x, \y) {\scriptsize $u_{\n}$}; \node[left] at (24 - \x, \y) {\scriptsize $v_{\n}$};}{}
		
		\ifthenelse{\n = 9 \OR \n=10 \OR \n=12 \OR \n=13 \OR \n=15 \OR \n=16 \OR \n=19} {\node[left] at (\x, \y) {\scriptsize $u_{\n}$}; \node[right] at (24 - \x, \y) {\scriptsize $v_{\n}$};}{}
		
		\ifthenelse{\n = 21} {\node[right] at (\x, \y) {\scriptsize $\botclausevar{a}{j}$}; \node[left] at (24 - \x, \y) {\scriptsize $\topclausevar{a}{j}$};}{}
		
		\ifthenelse{\n = 22} {\node[right] at (\x, \y) {\scriptsize $\topclausevar{c}{j}$}; \node[left] at (24 - \x, \y) {\scriptsize $\botclausevar{c}{j}$};}{}
	
	}
	
	\draw[dashed] (8.5, -5) rectangle (15.5,-3); 
	
	\foreach \x/\y [count = \n] in
  {9.15/-3, 12/-3, 14.85/-3}
  {
  	\node[above] at (\x,\y) {\scriptsize $D^{\n}_j$};
  }
	   \node[below] at (14.5,-5) {\scriptsize $D_j$};
	   
	   \foreach \x/\y/\w/\z [count = \n] in
  {9.15/-2.7/9.35/-4.15, 12/-2.7/12/-4.35, 14.85/-2.7/14.5/-4.15}
  {
  	\draw[densely dotted,->] (\x,\y) -- (\w,\z);
  }
  
  \foreach \x/\y [count = \n] in
  {8/-4.5, 8/-4, 8/-3.5}
  {

	\node[left] at (\x,\y) {\scriptsize $d_{\n}$};
  	
  }
  
  
%
\node[right] at (9,6) {\scriptsize $w_1$};
\node[right] at (10,6) {\scriptsize $w_2$};
\node[right] at (13,6) {\scriptsize $w_3$};
\node[right] at (14,6) {\scriptsize $w_4$};
\node[above] at (7,7) {\scriptsize $w_5$};
\node[above] at (17,7) {\scriptsize $w_6$};
\draw[dotted,->] (9.5,6.25) -- (8.25,6.75);
\draw[dotted,->] (10.25,6.25) -- (8.75, 7.25);
\draw[dotted,->] (13.5,6.25) -- (15.5,7.25);
\draw[dotted,->] (14.5,6.25) -- (15.75,6.75);

\foreach \x/\y [count = \n] in
  {7.5/14.25,7.5/12.25,7.5/10.25, 7.5/8.25, 7.5/-1.25, 7.5/-6.75}
  {

	\node[left] at (\x,\y) {\scriptsize $F_{\n}$};
  	\node[right] at (24-\x,\y) {\scriptsize $F'_{\n}$};
  }
  
  \foreach \x/\y/\w/\z [count = \n] in
  {7.25/14.25/8/13.5 ,7.25/12.25/8/11.5,7.25/10.25/8/9.5, 7.25/8.25/8/7.5, 7.25/-1.25/8/-2, 7.25/-6.75/8/-6}
  {

	\draw[densely dotted,->] (\x,\y) -- (\w,\z);
  	\draw[densely dotted,->] (24-\x,\y) -- (24-\w,\z);
  }
  
   \foreach \x/\y [count = \n] in
 {0/6, 1/6, 2/6, 3/6, 5/ 6, 6/ 6, 7/6, 8/6, 8/5, 8/4, 8/3, 8/2, 8/1, 8/0, 5/ 4.5, 6/ 4.5, 5/ 3, 6/ 3, 4/ 3, 6/1.5}
    {
    	\filldraw (\x, \y) circle (3pt); 
    	\filldraw (24 - \x, \y) circle (3pt); 
     
    }

    \foreach \x/\y [count = \n] in
 {8/13, 8.5/13.5, 8/14, 7.5/13.5, 
 8/11, 8.5/11.5, 8/12, 7.5/11.5, 
 8/9, 8.5/9.5, 8/10, 7.5/9.5, 
 8/7, 8.5/7.5, 8/8, 7.5/7.5,     
 8/-2.5, 8.5/ -2, 8/-1.5, 7.5/-2, 
 8/-6.5, 8.5/-6, 8/-5.5, 8/-4,7.5/-6, 
 9.5/-4.5, 10/ -4, 9.5/-3.5, 9/-3.5, 9/-4, 9/-4.5,
 9/-4, 11/-4, 11.5/-4.5, 11.5/-3.5 ,
 8/-3.5, 8/-4.5
 }
    {
    	\filldraw (\x, \y) circle (3pt);
        \filldraw (24-\x, \y) circle (3pt);
    }

    \end{tikzpicture}
    \caption{Clause gadget $C_j = \left( x_a,x_b,x_c \right)$ with $a<b<c$. 
    A red edge is selected in any perfect matching cut. 
    A blue edge is selected in some perfect matching cut. 
    A black edge is never selected in any perfect matching cut.}
    \label{fig:clause-gadget}
\end{figure}

\begin{enumerate}
\item For each variable $x_i$, let $\varGadget{i}$ denote a~fresh copy of the graph shown in~\Cref{fig:var-gadget}.
  Note that the variable $x_i$ appears in exactly four clauses, say, $C_j,C_k,C_p,C_q$ with $j<k<p<q$.
  The \emph{variable gadget} $\varGadget{i}$ contains the special vertices $\topvarclause{i}{j}$, $\botvarclause{i}{j}$, $\topvarclause{i}{k}$, $\botvarclause{i}{k}$, $\topvarclause{i}{p}$, $\botvarclause{i}{p}$, $\topvarclause{i}{q}$, $\botvarclause{i}{q}$ as shown in the figure.
  We recall that red edges are those forced in any perfect matching cut, while black edges cannot be in any solution.
  An essential part of the proof will consist of justifying the edge colors in our figures.
  
  For each clause $C_j = \left( x_a,x_b,x_c \right)$ with $a<b<c$ let $\clauseGadget{j}$ denote a new copy of the graph shown in \Cref{fig:clause-gadget}.
  The \emph{clause gadget} $\clauseGadget{j}$ contains the special vertices $\topclausevar{a}{j}$, $\botclausevar{a}{j}$, $\topclausevar{b}{j}$, $\botclausevar{b}{j}$, $\topclausevar{c}{j}$, $\botclausevar{c}{j}$, as shown in the figure. 
  Then for each variable $x_i$ that appears in the clause $C_j$, introduce two new edges $E_{ij} = \left\{ \topvarclause{i}{j} \topclausevar{i}{j}, \botvarclause{i}{j} \botclausevar{i}{j} \right\}$. 
  Let $\cubic$ denote the  graph defined as follows.
  $$V(\cubic) = \displaystyle\bigcup\limits_{i=1}^{n} V(\varGadget{i}) \cup \displaystyle\bigcup\limits_{j=1}^{m} V(\clauseGadget{j}) $$ $$ E(\cubic) = \displaystyle\bigcup\limits_{i=1}^{n} E(\varGadget{i}) \cup \displaystyle\bigcup\limits_{j=1}^{m} E(\clauseGadget{j}) \cup \displaystyle\bigcup\limits_{x_i\in C_j} E_{ij}.$$ 

  We assign to each edge $e \in E_{i,j}$ its variable as $\var(e) = i$.
  Note that, for a~variable gadget $\varGadget{i}$, there are exactly eight outgoing edges of $V(\varGadget{i})$.
  
  \medskip
   
    \item In the next step, we generate a drawing $\mathcal{R}$ of $\cubic$ on the plane according to the following procedure.
      \begin{enumerate}
      \item For each variable $x_i$, we embed $\varGadget{i}$ as a translate of the variable gadget of \Cref{fig:var-gadget} into $[0,1] \times [2i,2i + 1]$.
        \item For each clause $C_j$, we embed $\clauseGadget{j}$ as a translate of the clause gadget of \Cref{fig:clause-gadget} into $[2, 3] \times [2j, 2j+1]$.  
        
        \item Two edges incident to vertices in the same variable gadget or same clause gadget do not intersect in $\mathcal{R}$. For two variables $x_i, x_{i'}$ and clauses $C_j, C_{j'}$ with $x_i\in C_j, x_{i'} \in C_{j'}$, exactly one of the following holds:
        \begin{enumerate}
        \item For each pair of edges $(e,e') \in E_{ij} \times E_{i'j'}$, $e$ and $e'$ intersect exactly once in $\mathcal{R}$. When this condition is satisfied, we call $(E_{ij}, E_{i'j'})$ a~\emph{crossing quadruple}.
          Moreover, we ensure that the interior of the subsegment of $e \in E_{ij}$ between its two intersection points with edges of $E_{i'j'}$ is not crossed by any edge;
            \item There is no pair of edges $(e,e') \in E_{ij} \times E_{i'j'}$ such that $e$ and $e'$ intersect in $\mathcal{R}$; 
         \end{enumerate}
\end{enumerate}

    \medskip
     
    \item\label{it:cross} For each crossing quadruples $(E_{ij}, E_{i'j'})$ replace the four crossing points shown in \Cref{fig:cross-replace-a} by the crossing gadget shown in \Cref{fig:cross-replace-b}.
\end{enumerate}

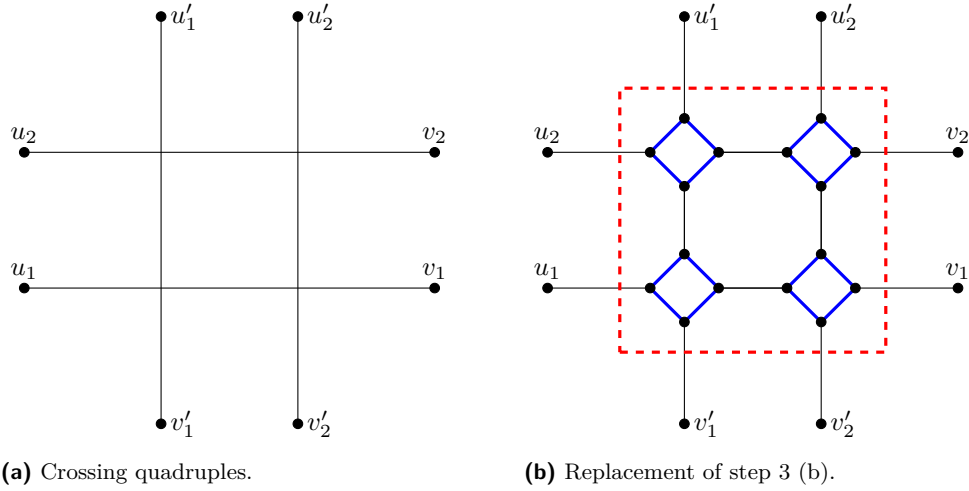
\begin{figure}[t]
\centering
\begin{subfigure}{0.45\textwidth}
\begin{tikzpicture}[scale=.9]
\foreach \x/\y/\z/\t [count = \n] in
 {0/2/6/2, 2/0/2/6, 0/4/6/4, 4/0/4/6} {
	\draw (\x, \y) -- (\z, \t);
}
\node[above] at (0, 2) {$u_1$};
\node[above] at (0, 4) {$u_2$};
\node[above] at (6, 2) {$v_1$};
\node[above] at (6, 4) {$v_2$};
\node[right] at (2, 6) {$u_1'$};
\node[right] at (4, 6) {$u_2'$};
\node[right] at (2, 0) {$v_1'$};
\node[right] at (4, 0) {$v_2'$};

\foreach \x/\y [count = \n] in
 {0/2, 0/4, 4/0, 2/0, 6/2, 6/4, 2/6, 4/6} {
	\filldraw (\x, \y) circle (2pt);
}	
\end{tikzpicture}
\subcaption{Crossing quadruples.}
\label{fig:cross-replace-a}
\end{subfigure}
~~~
\begin{subfigure}{0.45\textwidth}
\begin{tikzpicture}[scale=.9]

\foreach \x/\y/\z/\t [count = \n] in
 {1.5/2/2/1.5, 2/1.5/2.5/2, 2.5/2/2/2.5, 2/2.5/1.5/2,
 3.5/2/4/1.5, 4/1.5/4.5/2, 4.5/2/4/2.5, 4/2.5/3.5/2,
 1.5/4/2/3.5, 2/3.5/2.5/4, 2.5/4/2/4.5, 2/4.5/1.5/4,
 3.5/4/4/3.5, 4/3.5/4.5/4, 4.5/4/4/4.5, 4/4.5/3.5/4} {
	\draw[very thick, blue] (\x, \y) -- (\z, \t);
}
\node[above] at (0, 2) {$u_1$};
\node[above] at (0, 4) {$u_2$};
\node[above] at (6, 2) {$v_1$};
\node[above] at (6, 4) {$v_2$};
\node[right] at (2, 6) {$u_1'$};
\node[right] at (4, 6) {$u_2'$};
\node[right] at (2, 0) {$v_1'$};
\node[right] at (4, 0) {$v_2'$};

\foreach \x/\y/\z/\t [count = \n] in
 {0/2/1.5/2, 2.5/2/3.5/2, 2/0/2/1.5, 2/2.5/2/3.5,
  0/4/1.5/4, 2.5/4/3.5/4, 2/2.5/2/3.5, 2/4.5/2/6,
  2.5/4/3.5/4, 4.5/4/6/4, 4/2.5/4/3.5, 4/4.5/4/6,
  2.5/2/3.5/2, 4.5/2/6/2, 4/0/4/1.5, 4/2.5/4/3.5} {
	\draw (\x, \y) -- (\z, \t);
}

\foreach \x/\y [count = \n] in
 {0/2, 0/4, 2/0, 4/0, 6/2, 6/4, 2/6, 4/6, 1.5/2, 2/1.5, 2.5/2, 2/2.5, 3.5/2, 4/1.5, 4.5/2, 4/2.5, 1.5/4, 2/3.5, 2.5/4, 2/4.5, 3.5/4, 4/3.5, 4.5/4, 4/4.5} {
	\filldraw (\x, \y) circle (2pt);
}	
\node[rectangle,
    draw = red,
    very thick,
    dashed,
    minimum width = 3.5cm, 
    minimum height = 3.5cm] (r) at (3,3) {};
\end{tikzpicture} 
\subcaption{Replacement of step 3 (b).}
\label{fig:cross-replace-b}
\end{subfigure}
\caption{Replacement of a crossing by a crossing gadget.}
\label{fig:cross-replace}
\end{figure}
Let $\construct$ denote the resulting graph. We shall need the following definitions. 

\begin{definition}
Any edge of $\construct$ whose both endpoints are not contained withing the same gadget (variable, clause, or crossing) is a~\emph{connector edge}. 
Any endpoint of a connector edge is called a \emph{connector vertex}. 
For a connector edge $e$ incident to a crossing gadget, $var(e)$ is the index of the variable gadget it was originally going to. 
To each connector edge $uv$, we associate the variable $\var(uv)$ to both $u$ and $v$, denoted $\var(u), \var(v)$.
\end{definition}

\newcommand{\U}[1]{U_{j}}

\newcommand{\V}[1]{V_{j}}

Now we shall distinguish some 4-cycles of $\construct$.

\begin{definition}
An (induced) 4-cycle $C$ of $\construct$  is a \emph{crossover} 4-cycle if it belongs to some crossing gadget.
\end{definition}

 \begin{definition}
 An (induced) 4-cycle $C$ of $\construct$ is \emph{special} if $C$ it some $F_i$ or $F'_i$ of some~$\clauseGadget{j}$.
 \end{definition}

The special 4-cycles of a~particular clause gadget $\clauseGadget{j}$ are highlighted in~\Cref{fig:clause-gadget}.  
In the next section, we show that $\construct$ is indeed a~3-connected cubic bipartite planar graph.

\subsection{$\construct$ is Barnette}

We shall show that the constructed graph is Barnette.

\begin{lemma}\label{lem:3-connectedness}
The graph $\construct$ is 3-connected.
\end{lemma}
\begin{proof}
  Observe that, for any two adjacent gadgets $\mathcal X, \mathcal Y$, there are two disjoint connector edges from $\mathcal X$ to $\mathcal Y$.
  We consider $\construct$ after the removal of two vertices~$u, v$.

  First assume that $u$ and $v$ are not both connector vertices. 
  Thus two gadgets $\mathcal X$ and $\mathcal Y$ are adjacent in $\construct$ if and only if they are adjacent in $\construct - \{u, v\}$.
  And in particular, each pair of gadgets are then connected in $\construct - \{u, v\}$.

Then $\construct - \{u, v\}$ can only be disconnected if there exists, inside a~same gadget, two vertices that are disconnected in $\construct - \{u, v\}$.
In particular, this gadget is disconnected by the removal of $u, v$, which forces both $u$ and $v$ to be picked inside it.
Indeed every gadget is 2-connected.
We go through the three kinds of gadgets.
\begin{itemize}
\item If it is a variable gadget $\varGadget{i}$, then $\varGadget{i}$ is split in two components each containing a connector vertex.
  Let $a$ (resp.~$b$) be a~connector vertex in the first (resp.~second) component.
  As no ``variable'' vertex is a cutvertex in $\text{inc}(I)$, there is a~path $P$ in $\text{inc}(I)  - \{x_i\}$ from ``clause $C_{a'}$'' to ``clause $C_{b'}$'', where $C_{a'}$ and $C_{b'}$ are the clauses corresponding respectively to $a$ and $b$ (i.e., with the notations of~\Cref{fig:var-gadget}, $a', b'$ are the second indices of $a, b$ respectively).
Thus $a$ and $b$ are connected in $\construct - \{u, v\}$.
	\item If it is a clause gadget $\clauseGadget{j}$, it is split in two connected components such that one contains a connector vertex $t'_{i, j}$ and the other contains $b'_{i, j}$.
	Thus we can simply follow the path from $t'_{i, j}$ to $t_{i, j}$, then to $b_{i, j}$ and finally back to~$b'_{i, j}$ to connect the two parts of the clause gadget.  
      \item If it is a crossing gadget $X$, then the split separates $X$ in two connected components, but note that there exists a~gadget $\mathcal Y$ incident to both components.
        Thus, as $\mathcal Y$ is connected, the subgraph induced by their union, and hence $\construct - \{u, v\}$, is connected.
\end{itemize}

We now deal with the case when both $u$ and $v$ are connector vertices.
Observe that every gadget remains connected by removing up to two connector vertices inside it.
Therefore every gadget is connected in $\construct - \{u, v\}$.
By the first paragraph, the only interesting case is when $u$ and $v$ are the endpoints of two distinct connector edges between the same pair of gadgets.
Then, the effect of removing $u, v$ is to remove the link between the two gadgets.

However $\text{inc}(I)$ cannot have a bridge, for otherwise it would have a~``variable'' vertex that is a cutvertex.
In turn, one can see that this implies that the gadget adjacency graph is bridgeless.
\end{proof}

\begin{lemma}\label{lem:Barnette}
The graph $\construct$ is Barnette. 
\end{lemma}
\begin{proof}
  By the plane embedding of the crossing gadgets, $\construct$ is planar.
  One can check that $\construct$ is cubic, by observing that within each gadget (variable, clause, crossing), all the vertices have degree~3, except vertices of degree~2, 
  which are exactly those with an incident edge leaving the gadget.
  By \Cref{lem:3-connectedness}, $\construct$ is 3-connected.

  We shall thus prove the bipartiteness of $\construct$.
Recall that our construction had three main components: variable gadgets, clause gadgets and crossing gadgets. 
For a particular gadget $H$, observe that, all the outgoing edges of $H$ lie in the external face of $H$.
Circularly order the outgoing edges of $H$ by $e_1, \dots, e_p$, when going, say, clockwise. 
Take any two consecutive outgoing edges $e_i, e_{i+1}$.
Let $a_i(H), a_{i+1}(H)$ be the vertices of $H$ that are also incident to $e_i$ and $e_{i+1}$, respectively.
We can observe from our construction that the path from $a_i$ to $a_{i+1}$, 
denoted as $P(H,a_i,a_{i+1})$ along the external face of $H$ in clockwise order always has an even number of vertices. 

We call the path $P(H,a_i,a_{i+1})$ an \emph{exposed path} of $H$.
(Observe that a~particular gadget has several exposed paths.)
Let $F$ be a~bounded face of $\construct$.
If $F$ is a finite face of~a variable, clause, or crossing gadget, then $|V(F)|$ is even because $H$ is bipartite.
Otherwise, $F$~is a~union of exposed paths, and the previous arguments imply that $|V(F)|$ is even.
\end{proof}

\subsection{Properties of variable and crossing gadgets}

\begin{lemma}\label{lem:var-connector}
  Let $M$ be a perfect matching cut of $\construct$.
  Then for any variable gadget $\varGadget{i}$, $M \cap V(\varGadget{i})$ is the matching formed by the red edges in~\Cref{fig:var-gadget}.
  In particular, $M$ does not contain any connector edge incident to a~variable gadget. 
\end{lemma}
\begin{proof}
  Consider the variable gadget $\varGadget{i}$.
  By applying~\Cref{lem:hex-incident-square} on the 6-cycle~$\cycle{i}{2}$ (which satisfies the requirement of having four particular edges in some 4-cycles), we get that all outgoing edges of~$V(\cycle{i}{2})$ are in~$M$.
  We can thus apply \Cref{lem:hex-3-edges} on the 6-cycles~$\cycle{i}{1}$ and $\cycle{i}{3}$, and obtain that all outgoing edges of these cycles are in $M$.

  Now there is an outgoing edge of the 4-cycle $\cycle{i}{4}$ that is in $M$, hence by \Cref{lem:4-cycle}, all of them are.
  We can finally apply \Cref{lem:hex-3-edges} on the 6-cycle $\cycle{i}{5}$, and get that all the red edges of~\Cref{fig:var-gadget} should indeed be in~$M$.
  In particular, as all the vertices of $\varGadget{i}$ are touched by red edges, the connector edges incident to a~variable gadget cannot be in~$M$.
\end{proof}

Now we prove a property of the crossover 4-cycles.

\begin{lemma}\label{lem:F-cross}
  Let $M$ be a perfect matching cut of $\construct$ and $F$ be a crossover \mbox{4-cycle}.
  Then $\left| E(F) \right|=2$.
\end{lemma}
\begin{proof}
Say that a \emph{path of $4$-vertex cycles} is a sequence $C_1, \dots, C_k$ of vertex-disjoint $4$-cycles such that $C_i$ is adjacent to $C_{i+1}$.
Considering step 3 of the construction, observe that for every crossover 4-cycle $C$, there is a path of \mbox{4-vertex} cycles starting at $C$ and ending at a~crossover 4-cycle adjacent to a variable gadget.

By~\Cref{lem:var-connector}, no edge incident to a variable gadget is in~$M$.
Thus any crossover 4-cycle adjacent to a variable gadget contains an edge of~$M$.
Repeated applications of~\Cref{cor:4-cycle-propagation} imply that $C$ contains an edge of~$M$, and we conclude with~\Cref{lem:4-cycle} applied on~$C$. 
\end{proof}

\begin{corollary}\label{cor:connector}
For any perfect matching $M$ of $\construct$, $M$ contains no connector edges.
\end{corollary}
\begin{proof}
  We know that any connector edge incident to a~crossing gadget or to a~variable gadget is not in~$M$ by \Cref{lem:var-connector,lem:F-cross}.
\end{proof}

\subsection{Properties of clause gadgets}

Observe that $D_j$ is an induced subgraph of the variable gadget $\clauseGadget{j}$.

\begin{lemma}\label{lem:Dj-red-edges}
Any perfect matching cut of $\construct$ contains the edges of $D_j$ drawn in red in~\Cref{fig:clause-gadget}.
\end{lemma}
\begin{proof}
 Observe that the same proof for the variable gadgets already contained all the arguments. 
\end{proof}

We prove a property of the special 4-cycles of a clause gadget.

\begin{lemma}\label{lem:F-special}
  Let $M$ be a perfect matching cut of $\construct$ and $F$ be a special 4-cycle of $\clauseGadget{j}$.
  Then $\left| E(F) \cap M \right|=2$, and no outgoing edge of $V(F)$ is in $M$.
\end{lemma}

\begin{proof}
We know from \Cref{cor:connector} that a connector edge is not in $M$, and from \Cref{lem:Dj-red-edges} that $F_5$ (see \Cref{fig:clause-gadget}) has an incident edge not contained in $M$.
Thus every special $4$-cycle is connected by a path of $4$-cycles to a $4$-cycle incident to an edge not in $M$.
By application of \Cref{lem:4-cycle} and \Cref{cor:4-cycle-propagation}, every special $4$-cycle of $\clauseGadget{j}$ contains an edge of $M$.
\end{proof}

\begin{lemma}\label{lem:clause-property-1}
  Let $M$ be a perfect matching cut of $\construct$ and $\clauseGadget{j}$ be a clause gadget.
  Let $U_j = \{u_1, \dots, u_{20}\}$, and $V_j = \{v_1, \dots, v_{20}\}$.
  Then no outgoing edge of $U_j$ or of $V_j$ is in $M$.
\end{lemma}
\begin{proof}
By~\Cref{lem:F-special} and \Cref{cor:connector}, the only edges that remain to be checked are $u_9v_9, u_{10}v_{10}, u_{12}v_{12}, u_{13}v_{13}$.

Suppose $M$ contains $u_9v_9$. As $u_8u_9$ is not available, by \Cref{lem:4-cycle} on $u_7u_8w_1w_5$, $w_1w_2 \not \in M$. Symmetrically, $w_3w_4 \not \in M$.
As by \Cref{lem:F-special} $u_8w_1, v_8w_4$ and $w_2w_3$ are not in $M$, $u_9v_9$ would be the only edge of $u_9, u_8, w_1, w_2, w_3, w_4, v_8, v_9$ to be in $M$ which is absurd by \Cref{lem:pmc-cycles}.
A symmetric argument rules out that $u_{13}v_{13} \in M$.
Thus we conclude applying \Cref{lem:4-cycle} to the 4-cycles $u_9v_9v_{10}u_{10}$ and $u_{12}v_{12}v_{13}u_{13}$.
\end{proof}

From now on we assume that, for a clause gadget the two sets $U_j = \{u_1, \dots,$ $u_{20}\}$, and $V_j = \{v_1, \dots,$ $v_{20}\}$ are defined.
Now we shall prove that for every clause gadget $\clauseGadget{j}$ and perfect matching cut $M$, the set $M \cap E\left( \U{j} \cup \V{j} \right)$ can be of only three types.
Before we prove the corresponding lemma, we introduce the following notations. 
Let $H$ denote the subgraph of $\construct$ induced by the vertices of $\U{j} \cup \V{j}$ of the clause gadget $\clauseGadget{j}$. 
Let the vertices of $H$ be named as shown in~\Cref{fig:clause-gadget}.
We define the following sets:
\newcommand{\Type}[2]{{L}^{#1}_{#2}}
\newcommand{\Typecross}[2]{{P}^{#1}_{#2}}
\newcommand{\Typem}[2]{{R}^{#1}_{#2}}
$$ \Type{1}{j} = \left\{u_1u_2, u_3u_4, u_5u_{19}, u_6u_{20}, u_7u_8, u_9u_{10}, u_{16}u_{17}, u_{18}u_{11}, u_{12}u_{13}, u_{15}u_{14} \right\},$$ 
$$ \Type{2}{j} = \left\{u_{1}u_{2}, u_{3}u_{4}, u_{5}u_{6}, u_{7}u_{8}, u_{19}u_{20}, u_{9}u_{10}, u_{16}u_{15}, u_{17}u_{18},  u_{11}u_{12}, u_{14}u_{13} \right\},$$
$$ \Type{3}{j} = \left\{u_{2}u_{3}, u_{4}u_{5}, u_{6}u_{7}, u_{8}u_{9}, u_{1}u_{16}, u_{19}u_{17}, u_{20}u_{18}, u_{10}u_{11}, u_{12}u_{13},  u_{15}u_{14} \right\}.$$

For $i\in \{1,2,3\}$, let $\Typem{i}{j}$ denote the set of edges $\{v_k v_l \colon u_k u_l\in \Type{i}{j}\}$. 

\begin{definition}
We say that a perfect matching cut $M$ of $\construct$ is of \emph{type $i$ in $\clauseGadget{j}$} with  
$i\in \{1,2,3\}$, if $M\cap E(\U{j} \cup \V{j}) = \Type{i}{j} \cup \Typem{i}{j}$.
\end{definition}

\begin{figure}[t]
\centering
\begin{subfigure}{0.45\textwidth}
\begin{tikzpicture}
	 [scale = 0.5]    
    
 \foreach \x/\y [count = \n] in
 {0/6, 1/6, 2/6, 3/6, 5/ 6, 6/ 6, 7/6, 8/6, 8/5, 8/4, 8/3, 8/2, 8/1, 8/0, 6/1.5, 4/ 3, 5/ 3, 6/ 3, 5/ 4.5, 6/ 4.5 }
    {
    	\ifthenelse{\n<9}{  \ifthenelse{\n=1}
    								{   \node[left] at (\x, \y) {\scriptsize $u_{\n} = \scriptsize \topclausevar{b}{j}$};}
    	{ \ifthenelse{\n=8}{   \node[right] at (\x,\y ) {\scriptsize $u_{\n}$};}{ \node[above] at (\x, \y) 
        {\scriptsize $u_{\n}$};}}}
    	 { \ifthenelse{\n<15}{ \ifthenelse{\n=14}{   \node[below] at (\x,\y ) {\scriptsize $u_{\n}$};} 
    	 {  \node[right] at (\x, \y) {\scriptsize $u_{\n}$};}} 
    	 {  \ifthenelse{\n<17}{\node[left] at (\x, \y) 
    	 {\scriptsize $u_{\n}$};}{ \ifthenelse{\n<17}{\node[left] at (\x, \y) { \scriptsize $u_{\n}$};}{ \ifthenelse{\n<19}{\node[below] at (\x, \y) { \scriptsize $u_{\n}$};}
    	 { \ifthenelse{\n=19}{\node[left] at (\x, \y) { \scriptsize $u_{\n}$};}{ \ifthenelse{\n=20}{\node[right] at (\x, \y) { \scriptsize $u_{\n}$};}{   }  }  }  }  } }}
     
    }
    
\draw (0, 6)--(8, 6)--(8, 0) -- (0, 6);
\draw (4, 3) -- (8, 3);
\draw (5, 6)--(5, 3);
\draw (6, 6)--(6, 3);
\draw (5, 4.5)--(6, 4.5);

	
	\foreach \x/\y/\w/\z [count = \n] in
 {4/3/5/3, 6/3/8/3, 5/4.5/5/6, 6/4.5/6/6, 7/6/8/6, 8/5/8/4, 8/1/8/2, 8/0/6/1.5, 0/6/1/6, 2/6/3/6   }
    {
    	\draw[line width=0.07cm, brown] (\x, \y) -- (\w,\z); 
     
    }
	
     \foreach \x/\y [count = \n] in
 {0/6, 1/6, 2/6, 3/6, 5/ 6, 6/ 6, 7/6, 8/6, 8/5, 8/4, 8/3, 8/2, 8/1, 8/0, 5/ 4.5, 6/ 4.5, 5/ 3, 6/ 3, 4/ 3, 6/1.5}
    {
    	\filldraw (\x, \y) circle (3pt); 
     
    }
    
\end{tikzpicture}
\subcaption{Edges of $\Type{1}{j}$ are in brown.}
\label{fig:matching-types-a}
\end{subfigure}
~~~
\begin{subfigure}{0.45\textwidth}
\begin{tikzpicture}
	 [scale = 0.5]

 \foreach \x/\y [count = \n] in
 {0/6, 1/6, 2/6, 3/6, 5/ 6, 6/ 6, 7/6, 8/6, 8/5, 8/4, 8/3, 8/2, 8/1, 8/0, 6/1.5, 4/ 3, 5/ 3, 6/ 3, 5/ 4.5, 6/ 4.5 }
    {
    	\ifthenelse{\n<9}{  \ifthenelse{\n=1}
    								{   \node[left] at (\x, \y) {\scriptsize $u_{\n} = \scriptsize \topclausevar{b}{j}$};}
    	{ \ifthenelse{\n=8}{   \node[right] at (\x,\y ) {\scriptsize $u_{\n}$};}{ \node[above] at (\x, \y) 
        {\scriptsize $u_{\n}$};}}}
    	 { \ifthenelse{\n<15}{ \ifthenelse{\n=14}{   \node[below] at (\x,\y ) {\scriptsize $u_{\n}$};} 
    	 {  \node[right] at (\x, \y) {\scriptsize $u_{\n}$};}} 
    	 {  \ifthenelse{\n<17}{\node[left] at (\x, \y) 
    	 {\scriptsize $u_{\n}$};}{ \ifthenelse{\n<17}{\node[left] at (\x, \y) { \scriptsize $u_{\n}$};}{ \ifthenelse{\n<19}{\node[below] at (\x, \y) { \scriptsize $u_{\n}$};}
    	 { \ifthenelse{\n=19}{\node[left] at (\x, \y) { \scriptsize $u_{\n}$};}{ \ifthenelse{\n=20}{\node[right] at (\x, \y) { \scriptsize $u_{\n}$};}{   }  }  }  }  } }}
     
    }
    
\draw (0, 6)--(8, 6)--(8, 0) -- (0, 6);
\draw (4, 3) -- (8, 3);
\draw (5, 6)--(5, 3);
\draw (6, 6)--(6, 3);
\draw (5, 4.5)--(6, 4.5);
    
	
	\foreach \x/\y/\w/\z [count = \n] in
 {6/3/5/3, 6/4.5/5/4.5, 6/6/5/6, 7/6/8/6, 8/4/8/5, 8/3/8/2, 8/0/8/1, 6/1.5/4/3 , 0/6/1/6, 2/6/3/6  }
    {
    	\draw[line width=0.07cm, brown] (\x, \y) -- (\w,\z); 
     
    }    

	 \foreach \x/\y [count = \n] in
 {0/6, 1/6, 2/6, 3/6, 5/ 6, 6/ 6, 7/6, 8/6, 8/5, 8/4, 8/3, 8/2, 8/1, 8/0, 5/ 4.5, 6/ 4.5, 5/ 3, 6/ 3, 4/ 3, 6/1.5}
    {
    	\filldraw (\x, \y) circle (3pt); 
     
    }

\end{tikzpicture} 
\subcaption{Edges of $\Type{2}{j}$ are in brown.}
\label{fig:matching-types-b}
\end{subfigure}

\begin{subfigure}{0.45\textwidth}
\begin{tikzpicture}
	 [scale = 0.5]

 \foreach \x/\y [count = \n] in
 {0/6, 1/6, 2/6, 3/6, 5/ 6, 6/ 6, 7/6, 8/6, 8/5, 8/4, 8/3, 8/2, 8/1, 8/0, 6/1.5, 4/ 3, 5/ 3, 6/ 3, 5/ 4.5, 6/ 4.5 }
    {
    	\ifthenelse{\n<9}{  \ifthenelse{\n=1}
    								{   \node[left] at (\x, \y) {\scriptsize $u_{\n} = \scriptsize \topclausevar{b}{j}$};}
    	{ \ifthenelse{\n=8}{   \node[right] at (\x,\y ) {\scriptsize $u_{\n}$};}{ \node[above] at (\x, \y) 
        {\scriptsize $u_{\n}$};}}}
    	 { \ifthenelse{\n<15}{ \ifthenelse{\n=14}{   \node[below] at (\x,\y ) {\scriptsize $u_{\n}$};} 
    	 {  \node[right] at (\x, \y) {\scriptsize $u_{\n}$};}} 
    	 {  \ifthenelse{\n<17}{\node[left] at (\x, \y) 
    	 {\scriptsize $u_{\n}$};}{ \ifthenelse{\n<17}{\node[left] at (\x, \y) { \scriptsize $u_{\n}$};}{ \ifthenelse{\n<19}{\node[below] at (\x, \y) { \scriptsize $u_{\n}$};}
    	 { \ifthenelse{\n=19}{\node[left] at (\x, \y) { \scriptsize $u_{\n}$};}{ \ifthenelse{\n=20}{\node[right] at (\x, \y) { \scriptsize $u_{\n}$};}{   }  }  }  }  } }}
     
    }
    
\draw (0, 6)--(8, 6)--(8, 0) -- (0, 6);
\draw (4, 3) -- (8, 3);
\draw (5, 6)--(5, 3);
\draw (6, 6)--(6, 3);
\draw (5, 4.5)--(6, 4.5);
    
	
	\foreach \x/\y/\w/\z [count = \n] in
 {6/3/6/4.5, 5/3/5/4.5, 3/6/5/6, 7/6/6/6, 8/6/8/5, 8/3/8/4, 8/2/8/1, 6/1.5/8/0 , 0/6/4/3, 2/6/1/6  }
    {
    	\draw[line width=0.07cm, brown] (\x, \y) -- (\w,\z); 
     
    }    
    
    \foreach \x/\y [count = \n] in
 {0/6, 1/6, 2/6, 3/6, 5/ 6, 6/ 6, 7/6, 8/6, 8/5, 8/4, 8/3, 8/2, 8/1, 8/0, 5/ 4.5, 6/ 4.5, 5/ 3, 6/ 3, 4/ 3, 6/1.5}
    {
    	\filldraw (\x, \y) circle (3pt); 
     
    }
    
\end{tikzpicture}
\caption{Edges of $\Type{3}{j}$ are in brown.}
\label{fig:matching-types-c}
\end{subfigure}
\caption{The three types of perfect matching cuts within a clause gadget.}
\label{fig:matching-types}
\end{figure}
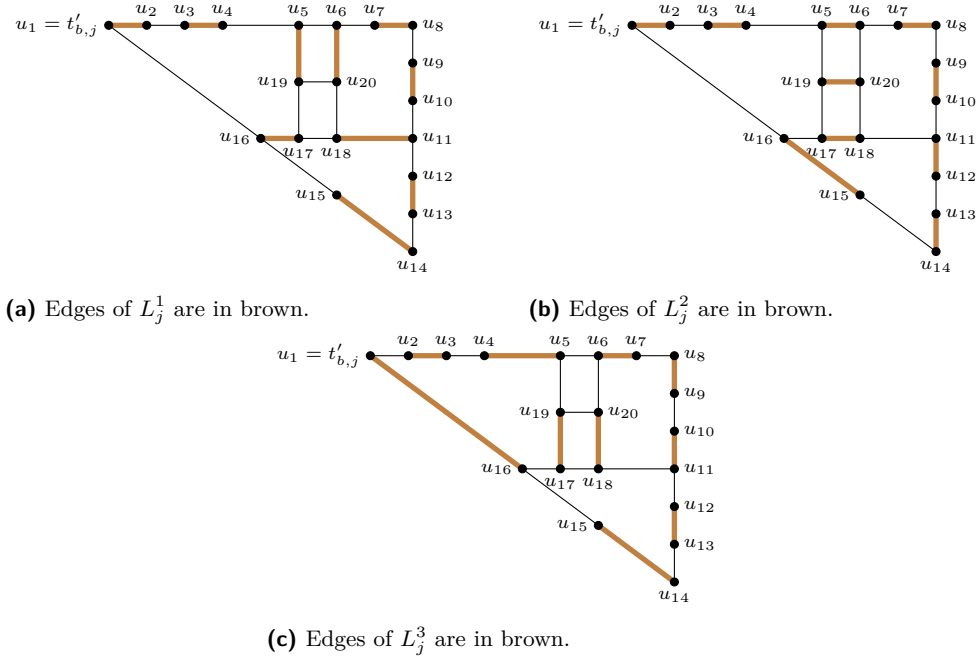

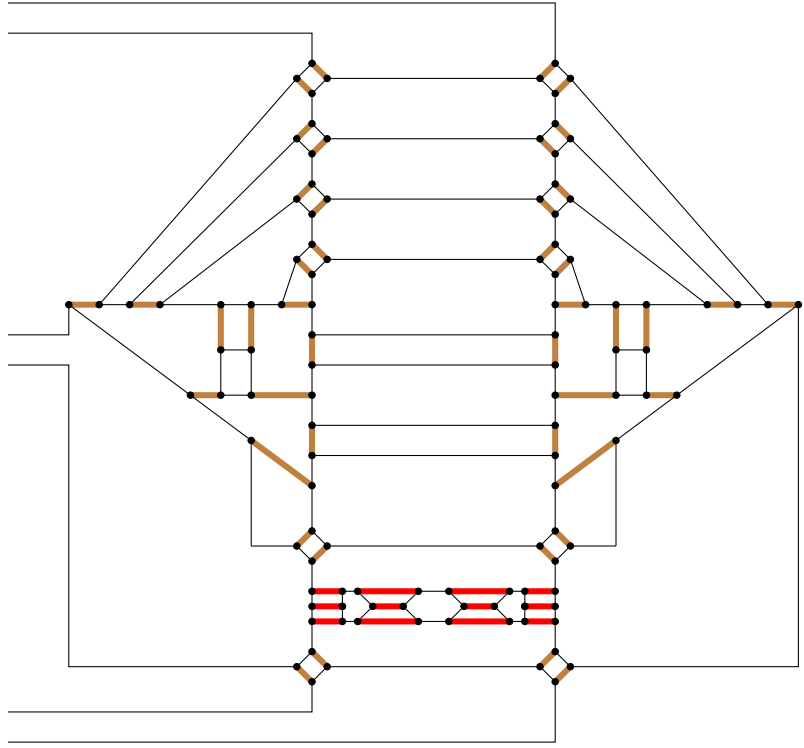
\begin{figure}[t]
    \centering
        \begin{tikzpicture}[scale = 0.4]

        \foreach \x/\y/\w/\z [count = \n] in
 {4/3/5/3, 6/3/8/3, 5/4.5/5/6, 6/4.5/6/6, 7/6/8/6, 8/5/8/4, 8/1/8/2, 8/0/6/1.5, 0/6/1/6, 2/6/3/6   }
    {
    	\draw[line width=0.08cm, brown] (\x, \y) -- (\w,\z); 
     	\draw[line width=0.08cm, brown] (24-\x, \y) -- (24-\w,\z); 
    }
    \foreach \x/\y/\w/\z [count = \n] in
 {1/6/2/6, 3/6/7/6, 8/6/8/5, 8/4/8/2, 8/1/8/0, 0/6/6/1.5, 5/4.5/5/3, 6/4.5/6/3, 5/4.5/6/4.5, 5/3/6/3  }
    {
    	\draw (\x,\y) -- (\w,\z); 
    	\draw (24 - \x, \y) -- (24 - \w, \z); 
     
    }
    
    \foreach \x/\y [count = \n] in
 {8/5, 8/4, 8/2, 8/1, 8.5/-2, 8.5/-6, 11.5/-3.5, 11.5/-4.5, 8.5/7.5, 8.5/9.5, 8.5/11.5, 8.5/13.5}
    {
    	\draw (\x,\y) -- (24-\x,\y); 
     
    }

\foreach \x/\y/\w/\z [count = \n] in
 {
8.5/13.5/8/14, 7.5/13.5/8/13,
  8/11/8.5/11.5, 8/12/7.5/11.5,
  8/9/8.5/9.5, 8/10/7.5/9.5,
  8.5/7.5/8/8, 7.5/7.5/8/7,
  8/-2.5/8.5/-2, 8/-1.5/7.5/-2,
  8.5/-6/8/-5.5, 7.5/-6/8/-6.5
   }
    {
    	
    	\draw[line width=0.08cm, brown] (\x,\y) -- (\w,\z); 
        \draw[line width=0.08cm, brown] (24-\x,\y) -- (24-\w,\z);
    }
    
    \foreach \x/\y/\w/\z [count = \n] in
 {
8/13/8.5/13.5, 8/14/7.5/13.5, 
  8.5/11.5/8/12, 7.5/11.5/8/11,
  8.5/9.5/8/10, 7.5/9.5/8/9,
  8/7/8.5/7.5, 8/8/7.5/7.5,
  8.5/ -2/8/-1.5, 7.5/-2/8/-2.5,
 8/-6.5/8.5/-6, 8/-5.5/7.5/-6
   }
    {
    	
    	\draw (\x,\y) -- (\w,\z); 
        \draw (24-\x,\y) -- (24-\w,\z);
    }

    \foreach \x/\y/\w/\z [count = \n] in
 { 8/13/8/12, 8/11/8/10, 8/9/8/8, 8/7/8/6, 
 8/-1.5/8/0, 8/-2.5/8/-5.5, 7.5/-2/6/-2, 
 6/-2/6/1.5,  8/-4/9/-4, 9.5/-4.5/10/ -4, 
 10/ -4/9.5/-3.5, 9.5/-3.5/9/-3.5, 9/-3.5/9/-4, 9/-4/9/-4.5, 9/-4.5/10/-4.5, 
 11/-4/11.5/-4.5, 11/-4/11.5/-3.5, 
 7.5/7.5/7/6, 7.5/9.5/3/6, 7.5/11.5/2/6, 7.5/13.5/1/6 }
    {
    	\draw (\x,\y) -- (\w,\z); 
        \draw (24-\x,\y) -- (24-\w,\z);
    }    
    
    \foreach \x/\y/\w/\z [count = \n] in
 { 8/-4/9/-4, 10/-4/11/-4, 9.5/-4.5/11.5/-4.5, 9.5/-3.5/11.5/-3.5,
 9/-4.5/8/-4.5, 9/-3.5/8/-3.5}
    {
    	\draw[line width=0.08cm, red] (\x,\y) -- (\w,\z); 
        \draw[line width=0.08cm, red] (24-\x,\y) -- (24-\w,\z);
    }

    
    \draw (24,6) -- (24,-6) -- (16.5,-6);
    \draw (0,6) -- (0,5) -- (-2,5); 
    \draw (7.5,-6) -- (0,-6) -- (0,4) -- (-2,4);
    
    \draw (16,-6.5) -- (16,-8.5) -- (-2,-8.5);
    \draw (8,-6.5) -- (8,-7.5) -- (-2,-7.5);
    
     \draw (16,14) -- (16,16) -- (-2,16);
    \draw (8,14) -- (8,15) -- (-2,15);
  
      \foreach \x/\y [count = \n] in
 {8/13, 8.5/13.5, 8/14, 7.5/13.5, 
 8/11, 8.5/11.5, 8/12, 7.5/11.5, 
 8/9, 8.5/9.5, 8/10, 7.5/9.5, 
 8/7, 8.5/7.5, 8/8, 7.5/7.5,     
 8/-2.5, 8.5/ -2, 8/-1.5, 7.5/-2, 
 8/-6.5, 8.5/-6, 8/-5.5, 8/-4,7.5/-6, 
 9.5/-4.5, 10/ -4, 9.5/-3.5, 9/-3.5, 9/-4, 9/-4.5,
 9/-4, 11/-4, 11.5/-4.5, 11.5/-3.5 ,
 8/-3.5, 8/-4.5
 }
    {
    	\filldraw (\x, \y) circle (3pt);
        \filldraw (24-\x, \y) circle (3pt);
    }

 \foreach \x/\y [count = \n] in
 {0/6, 1/6, 2/6, 3/6, 5/ 6, 6/ 6, 7/6, 8/6, 8/5, 8/4, 8/3, 8/2, 8/1, 8/0, 5/ 4.5, 6/ 4.5, 5/ 3, 6/ 3, 4/ 3, 6/1.5}
    {
    	\filldraw (\x, \y) circle (3pt); 
    	\filldraw (24 - \x, \y) circle (3pt); 
     
    }
    
    \end{tikzpicture}
    \caption{The brown and red edges make the only intersection of a~clause gadget with a~perfect matching cut $M$ such that $M \cap U_j = \Type{1}{j}$.}
    \label{fig:clause-gadget-type-1-match}
\end{figure}

\begin{figure}[t]
  \centering
  \begin{subfigure}{0.48\textwidth}
        \begin{tikzpicture}[scale = 0.22]
        
        \foreach \x/\y/\w/\z [count = \n] in
 {6/3/5/3, 6/4.5/5/4.5, 6/6/5/6, 7/6/8/6, 8/4/8/5, 8/3/8/2, 8/0/8/1, 6/1.5/4/3 , 0/6/1/6, 2/6/3/6  }
    {
    	\draw[line width=0.06cm, brown] (\x, \y) -- (\w,\z); 
     	\draw[line width=0.06cm, brown] (24-\x, \y) -- (24-\w,\z); 
    } 

    \foreach \x/\y/\w/\z [count = \n] in
 {1/6/2/6, 3/6/5/6, 6/6/7/6, 8/6/8/5, 8/4/8/3, 8/2/8/1, 8/0/6/1.5, 0/6/4/3, 5/6/5/3, 6/6/6/3, 4/3/5/3, 8/3/6/3  }
    {
    	\draw (\x,\y) -- (\w,\z); 
    	\draw (24 - \x, \y) -- (24 - \w, \z); 
     
    }
    
    \foreach \x/\y [count = \n] in
 {8/5, 8/4, 8/2, 8/1, 8.5/-2, 8.5/-6, 11.5/-3.5, 11.5/-4.5, 8.5/7.5, 8.5/9.5, 8.5/11.5, 8.5/13.5}
    {
    	\draw (\x,\y) -- (24-\x,\y); 
     
    }

\foreach \x/\y/\w/\z [count = \n] in
 {
8/13/8.5/13.5, 8/14/7.5/13.5,
    8.5/11.5/8/12, 7.5/11.5/8/11,
  8.5/9.5/8/10, 7.5/9.5/8/9,
  8.5/7.5/8/8, 7.5/7.5/8/7,
8.5/ -2/8/-1.5, 7.5/-2/8/-2.5,
  8.5/-6/8/-5.5, 7.5/-6/8/-6.5
   }
    {
    	\draw[line width=0.06cm, brown] (\x,\y) -- (\w,\z); 
        \draw[line width=0.06cm, brown] (24-\x,\y) -- (24-\w,\z);
    }

    \foreach \x/\y/\w/\z [count = \n] in
 {
8.5/13.5/8/14, 7.5/13.5/8/13, 
8/11/8.5/11.5, 8/12/7.5/11.5,
  8/9/8.5/9.5, 8/10/7.5/9.5,
  8/7/8.5/7.5, 8/8/7.5/7.5,
  8/-2.5/8.5/-2, 8/-1.5/7.5/-2,
 8/-6.5/8.5/-6, 8/-5.5/7.5/-6
   }
    {
    	
    	\draw (\x,\y) -- (\w,\z); 
        \draw (24-\x,\y) -- (24-\w,\z);
    }

    \foreach \x/\y/\w/\z [count = \n] in
 { 8/13/8/12, 8/11/8/10, 8/9/8/8, 8/7/8/6, 
 8/-1.5/8/0, 8/-2.5/8/-5.5, 7.5/-2/6/-2, 
 6/-2/6/1.5,  8/-4/9/-4, 9.5/-4.5/10/ -4, 
 10/ -4/9.5/-3.5, 9.5/-3.5/9/-3.5, 9/-3.5/9/-4, 9/-4/9/-4.5, 9/-4.5/10/-4.5, 
 11/-4/11.5/-4.5, 11/-4/11.5/-3.5, 
 7.5/7.5/7/6, 7.5/9.5/3/6, 7.5/11.5/2/6, 7.5/13.5/1/6 }
    {
    	\draw (\x,\y) -- (\w,\z); 
        \draw (24-\x,\y) -- (24-\w,\z);
    }    
    
    \foreach \x/\y/\w/\z [count = \n] in
 { 8/-4/9/-4, 10/-4/11/-4, 9.5/-4.5/11.5/-4.5, 9.5/-3.5/11.5/-3.5,
 9/-4.5/8/-4.5, 9/-3.5/8/-3.5}
    {
    	\draw[line width=0.06cm, red] (\x,\y) -- (\w,\z); 
        \draw[line width=0.06cm, red] (24-\x,\y) -- (24-\w,\z);
    }

    
    \draw (24,6) -- (24,-6) -- (16.5,-6);
    \draw (0,6) -- (0,5) -- (-2,5); 
    \draw (7.5,-6) -- (0,-6) -- (0,4) -- (-2,4);
    
    \draw (16,-6.5) -- (16,-8.5) -- (-2,-8.5);
    \draw (8,-6.5) -- (8,-7.5) -- (-2,-7.5);
    
     \draw (16,14) -- (16,16) -- (-2,16);
    \draw (8,14) -- (8,15) -- (-2,15);
  
      \foreach \x/\y [count = \n] in
 {8/13, 8.5/13.5, 8/14, 7.5/13.5, 
 8/11, 8.5/11.5, 8/12, 7.5/11.5, 
 8/9, 8.5/9.5, 8/10, 7.5/9.5, 
 8/7, 8.5/7.5, 8/8, 7.5/7.5,     
 8/-2.5, 8.5/ -2, 8/-1.5, 7.5/-2, 
 8/-6.5, 8.5/-6, 8/-5.5, 8/-4,7.5/-6, 
 9.5/-4.5, 10/ -4, 9.5/-3.5, 9/-3.5, 9/-4, 9/-4.5,
 9/-4, 11/-4, 11.5/-4.5, 11.5/-3.5 ,
 8/-3.5, 8/-4.5
 }
    {
    	\filldraw (\x, \y) circle (3pt);
        \filldraw (24-\x, \y) circle (3pt);
    }  
    
 \foreach \x/\y [count = \n] in
 {0/6, 1/6, 2/6, 3/6, 5/ 6, 6/ 6, 7/6, 8/6, 8/5, 8/4, 8/3, 8/2, 8/1, 8/0, 5/ 4.5, 6/ 4.5, 5/ 3, 6/ 3, 4/ 3, 6/1.5}
    {
    	\filldraw (\x, \y) circle (3pt); 
    	\filldraw (24 - \x, \y) circle (3pt); 
     
    }
    
    \end{tikzpicture}
    \caption{Edges (brown and red) implied by $\Type{2}{j}$.}
    \label{fig:clause-gadget-type-2-match}
  \end{subfigure}
  ~
\begin{subfigure}{0.48\textwidth}
    \centering
        \begin{tikzpicture}[scale = 0.22]
   
   	\foreach \x/\y/\w/\z [count = \n] in
 {6/3/6/4.5, 5/3/5/4.5, 3/6/5/6, 7/6/6/6, 8/6/8/5, 8/3/8/4, 8/2/8/1, 6/1.5/8/0 , 0/6/4/3, 2/6/1/6  }
    {
    	\draw[line width=0.06cm, brown] (\x, \y) -- (\w,\z); 
     	\draw[line width=0.06cm, brown] (24-\x, \y) -- (24-\w,\z); 
    } 

    \foreach \x/\y/\w/\z [count = \n] in
 {0/6/1/6, 3/6/2/6, 6/6/5/6, 7/6/8/6, 8/4/8/5, 8/2/8/3, 8/0/8/1, 4/3/6/1.5, 8/3/4/3, 5/6/5/4.5, 6/6/6/4.5, 5/4.5/6/4.5  }
    {
    	\draw (\x,\y) -- (\w,\z); 
    	\draw (24 - \x, \y) -- (24 - \w, \z); 
     
    }

    \foreach \x/\y [count = \n] in
 {8/5, 8/4, 8/2, 8/1, 8.5/-2, 8.5/-6, 11.5/-3.5, 11.5/-4.5, 8.5/7.5, 8.5/9.5, 8.5/11.5, 8.5/13.5}
    {
    	\draw (\x,\y) -- (24-\x,\y); 
     
    }

\foreach \x/\y/\w/\z [count = \n] in
 {
8/13/8.5/13.5, 8/14/7.5/13.5,
    8/11/8.5/11.5, 8/12/7.5/11.5,
  8.5/9.5/8/10, 7.5/9.5/8/9,
  8/7/8.5/7.5, 8/8/7.5/7.5,
8/-2.5/8.5/-2, 8/-1.5/7.5/-2,
  8/-6.5/8.5/-6, 8/-5.5/7.5/-6
   }
    {
    	\draw[line width=0.06cm,brown] (\x,\y) -- (\w,\z); 
        \draw[line width=0.06cm,brown] (24-\x,\y) -- (24-\w,\z);
    }

    \foreach \x/\y/\w/\z [count = \n] in
 {
8.5/13.5/8/14, 7.5/13.5/8/13, 
8.5/11.5/8/12, 7.5/11.5/8/11,
  8/9/8.5/9.5, 8/10/7.5/9.5,
  8.5/7.5/8/8, 7.5/7.5/8/7,
  8.5/ -2/8/-1.5, 7.5/-2/8/-2.5,
 8.5/-6/8/-5.5, 7.5/-6/8/-6.5
   }
    {
    	
    	\draw (\x,\y) -- (\w,\z); 
        \draw (24-\x,\y) -- (24-\w,\z);
    }

    \foreach \x/\y/\w/\z [count = \n] in
 { 8/13/8/12, 8/11/8/10, 8/9/8/8, 8/7/8/6, 
 8/-1.5/8/0, 8/-2.5/8/-5.5, 7.5/-2/6/-2, 
 6/-2/6/1.5,  8/-4/9/-4, 9.5/-4.5/10/ -4, 
 10/ -4/9.5/-3.5, 9.5/-3.5/9/-3.5, 9/-3.5/9/-4, 9/-4/9/-4.5, 9/-4.5/10/-4.5, 
 11/-4/11.5/-4.5, 11/-4/11.5/-3.5, 
 7.5/7.5/7/6, 7.5/9.5/3/6, 7.5/11.5/2/6, 7.5/13.5/1/6 }
    {
    	\draw (\x,\y) -- (\w,\z); 
        \draw (24-\x,\y) -- (24-\w,\z);
    }    
    
    \foreach \x/\y/\w/\z [count = \n] in
 { 8/-4/9/-4, 10/-4/11/-4, 9.5/-4.5/11.5/-4.5, 9.5/-3.5/11.5/-3.5,
 9/-4.5/8/-4.5, 9/-3.5/8/-3.5}
    {
    	\draw[line width=0.06cm, red] (\x,\y) -- (\w,\z); 
        \draw[line width=0.06cm, red] (24-\x,\y) -- (24-\w,\z);
    }

    
    \draw (24,6) -- (24,-6) -- (16.5,-6);
    \draw (0,6) -- (0,5) -- (-2,5); 
    \draw (7.5,-6) -- (0,-6) -- (0,4) -- (-2,4);
    
    \draw (16,-6.5) -- (16,-8.5) -- (-2,-8.5);
    \draw (8,-6.5) -- (8,-7.5) -- (-2,-7.5);
    
     \draw (16,14) -- (16,16) -- (-2,16);
    \draw (8,14) -- (8,15) -- (-2,15);
  
      \foreach \x/\y [count = \n] in
 {8/13, 8.5/13.5, 8/14, 7.5/13.5, 
 8/11, 8.5/11.5, 8/12, 7.5/11.5, 
 8/9, 8.5/9.5, 8/10, 7.5/9.5, 
 8/7, 8.5/7.5, 8/8, 7.5/7.5,     
 8/-2.5, 8.5/ -2, 8/-1.5, 7.5/-2, 
 8/-6.5, 8.5/-6, 8/-5.5, 8/-4,7.5/-6, 
 9.5/-4.5, 10/ -4, 9.5/-3.5, 9/-3.5, 9/-4, 9/-4.5,
 9/-4, 11/-4, 11.5/-4.5, 11.5/-3.5 ,
 8/-3.5, 8/-4.5
 }
    {
    	\filldraw (\x, \y) circle (3pt);
        \filldraw (24-\x, \y) circle (3pt);
    }  
    
 \foreach \x/\y [count = \n] in
 {0/6, 1/6, 2/6, 3/6, 5/ 6, 6/ 6, 7/6, 8/6, 8/5, 8/4, 8/3, 8/2, 8/1, 8/0, 5/ 4.5, 6/ 4.5, 5/ 3, 6/ 3, 4/ 3, 6/1.5}
    {
    	\filldraw (\x, \y) circle (3pt); 
    	\filldraw (24 - \x, \y) circle (3pt); 
     
    }
    
    \end{tikzpicture}
    \caption{Edges (brown and red) implied by $\Type{3}{j}$.}
    \label{fig:clause-gadget-type-3-match}
\end{subfigure}
\caption{Same as~\Cref{fig:clause-gadget-type-1-match} for $\Type{2}{j}$ (left) and $\Type{3}{j}$ (right).}
\label{fig::clause-gadget-type-23-match}
\end{figure}

\begin{lemma}\label{lem:matching-type}
Let $M$ be a perfect matching cut of $\construct$ and $\clauseGadget{j}$ be a clause gadget. 
Then there exists exactly one integer $i\in \{1,2,3\}$ such that $M$ is of type $i$ in $\clauseGadget{j}$.
\end{lemma}
\begin{proof}
Let $H$ denote the subgraph of $\construct$ induced by the vertices of $\U{j} \cup \V{j}$ of the clause gadget~$\clauseGadget{j}$. 
Let $F = \{u_9v_9, u_{10}v_{10}, u_{12}v_{12}, u_{13}v_{13}\}$. 
Consider the 4-cycle~$C$ induced by $u_{17}, u_{18}, u_{19}, u_{20}$. 
Consider the case when $M \cap E(C) = \emptyset$.
In this case, applying \Cref{lem:4-cycle} on~$C$, we know that $\{u_{19}u_5, u_{20}u_6, u_{17}u_{16}, u_{18}u_{11} \} \subset M$; see \Cref{fig:matching-types-a}.
Since no outgoing edge of $U_j$ is in $M$, due to \Cref{lem:clause-property-1}, it is now easy to verify that $\Type{1}{j} \subset M$. 
\sloppy In the case where $M \cap E(C) = \{ u_{19}u_{20}, u_{17}u_{18}\}$, applying \Cref{lem:4-cycle} on the 4-cycle induced by $u_5, u_6, u_{19}, u_{20}$, we infer that $u_5u_{6}\in M$, and once again it is now easy to verify that $\Type{2}{j} \subset M$; see \Cref{fig:matching-types-b}. 
\sloppy In the last case, $M \cap E(C) = \{ u_{19}u_{17}, u_{18}u_{20}\}$.
We again apply \Cref{lem:4-cycle} on the 4-cycle induced by $u_5, u_6, u_{19}, u_{20}$, and infer this time that $u_5u_{6} \not \in M$. 
As no outgoing edge of $U_j$ is in $M$, it is now easy to verify that $\Type{3}{j} \subset M$; see \Cref{fig:matching-types-c}. 

Observe that in $\Type{1}{j}$, $u_9u_{10} \in M$ and $u_{12}u_{13} \in M$, while for $\Type{2}{j}$ we have $u_9u_{10} \in M$ and $u_{12}u_{13} \not \in M$ and for $\Type{2}{j}$ we have $u_9u_{10} \not \in M$ and $u_{12}u_{13} \in M$.
Thus $M \cap U_j$ is determined by the containment of $u_9u_{10}$ and of $u_{12}u_{13}$ in $M$.
This is also the fact, by symmetry, for $V_j \cap M$, when considering the edges $v_9v_{10}$ and $v_{12}v_{13}$.

At this point, apply \Cref{lem:4-cycle} to the two 4-cycles $u_9, u_{10}, v_{10}, v_{9}$ and $v_{17}, v_{18}, v_{19}, v_{20}$. We have that $u_9u_{10} \in M$ if and only if $v_9v_{10}\in M$, and $u_{12}u_{13} \in M$ if and only if $v_{12}v_{13}\in M$. Thus $\Type{i}{j}$ propagates to $\Type{i}{j} \cup \Typem{i}{j}$.
\end{proof}

As a direct consequence of \Cref{lem:matching-type}, we get the following.

\begin{lemma}\label{lem:corner-not-same-side}
Let $M$ be a perfect matching cut of $\construct$ and $(A,B)$ be the cut of~$M$. 
The vertices $u_1,u_8,u_{14}$ of a clause gadget $\clauseGadget{j}$ cannot all be on the same side of~$M$.
More precisely:
\begin{enumerate}
	\item $\Type{1}{j}$ sets $u_1$ to one side of $M$, and $u_8, u_{14}$ to the other;
	\item $\Type{2}{j}$ sets $u_{14}$ to one side of $M$, and $u_1, u_{8}$ to the other;
	\item $\Type{3}{j}$ sets $u_8$ to one side of $M$, and $u_1, u_{14}$ to the other.
\end{enumerate}
\end{lemma}

Note that for a clause gadget $\clauseGadget{j}$, if $M$ is of type 1 (type 2, type 3, respectively) in $\clauseGadget{j}$, then the edges in $M \cap E\left( \clauseGadget{j}\right)$ are indicated in \Cref{fig:clause-gadget-type-1-match} (\Cref{fig:clause-gadget-type-2-match}, \Cref{fig:clause-gadget-type-3-match}, respectively).

\subsection{Relation between variable and clause gadgets}

\begin{lemma}\label{lem:variable-clause-propagation}
Let $M$ be a perfect matching cut of $\construct$.
Then for a variable $x_i$ and a clause $C_j$ with $x_i\in C_j$, $\topvarclause{i}{j}, \topclausevar{i}{j}, \botvarclause{i}{j}, \botclausevar{i}{j}$ are on the same side of~$M$.
\end{lemma}

\begin{proof}
Let $z$ be any vertex of the cycle $\cycle{i}{2}$ in the variable gadget $\varGadget{i}$; see \Cref{fig:var-gadget}.
Observe that there exists a path $P$ (resp. $P'$) between $z$ and $\topvarclause{i}{j}$ (resp. $\botvarclause{i}{j}$) such that $|M  \cap E(P)|$ (resp. $|M  \cap E(P')|$) is even.
Hence, due to \Cref{obs:sides}, $\topvarclause{i}{j}$ and $\botvarclause{i}{j}$ are on the same side of~$M$.

Our construction of $\construct$ ensures that there exists an even non-negative integer $k$ (where $k = 0$ if $\topvarclause{i}{j}$ and $\topclausevar{i}{j}$ are adjacent) such that all the following holds:
\begin{itemize}
\item there are $k$ crossover 4-cycles $F_1, F_2, \ldots, F_k$ and a path $P$ between $\topvarclause{i}{j}$ and $\topclausevar{i}{j}$ where 
$$V(P) \setminus \{ \topvarclause{i}{j}, \topclausevar{i}{j} \} \subset \displaystyle\bigcup_{l\in k} V(F_l)$$
\item for each $1\leq l\leq k$, $E(P) \cap E(F_l)$ is a~2-edge subpath.
\end{itemize} 

Now due to \Cref{lem:F-cross} we know that for any $1\leq l\leq k$, $|M \cap E(F_l)|=2$. 
The above arguments further imply that $|M \cap E(F_l) \cap E(P)|=1$. This implies that $|E(P) \cap M|=k$, which is even. 
Hence due to \Cref{obs:sides} we have that $\topvarclause{i}{j}$ and $\topclausevar{i}{j}$ are on the same side of~$M$. 

Using similar reasoning we can infer that $\botclausevar{i}{j}$ is on the same side as $\botvarclause{i}{j}$. Hence  
$\topvarclause{i}{j}, \topclausevar{i}{j}, \botvarclause{i}{j}, \botclausevar{i}{j}$ are all on the same side. 
\end{proof}

\begin{lemma}\label{lem:intra-clause-propagation}
Let $M$ be a perfect matching cut of $\construct$.
Then for any clause gadget $\clauseGadget{j}$ corresponding to the clause $C_j=(x_a,x_b,x_c)$ with $a<b<c$, the following hold: 
\begin{enumerate}[label=(\alph*)]
\item $\topclausevar{c}{j}$ and $u_{14}$ are on the same side of $M$, and
\item $\botclausevar{a}{j}$ and $u_{8}$ are on the same side of $M$. 
\end{enumerate} 
\end{lemma}

\begin{proof}
First we prove $(a)$. 
Using \Cref{fig:clause-gadget} observe that there exists a path $P$ between $\topclausevar{c}{j}$ and $u_{14}$ such that $P$ can be written as 
$\topclausevar{c}{j}~z_1~z_2~d_1~d_2~d_3~z_3~z_4~z_5~u_{14}$ where $\{z_1,z_2\}\subset V(F_6)$ and $\{z_3,z_4,z_5\}\subset V(F_5)$. 
Note that $F_5$ and $F_6$ are special 4-cycles. 
Due to \Cref{lem:F-special}, we have that $|M \cap E(F_5)| = 2$ and $|M \cap E(F_6)| = 2$. 
This implies there exists exactly one edge $e\in \{\topclausevar{c}{j}z_1, z_1z_2\}$ such that \mbox{$e\in M$}. 
Similarly, there exists exactly one edge $e'\in \{z_3z_4, z_4z_5\}$ such that \mbox{$e'\in M$}.
Moreover, from \Cref{lem:Dj-red-edges} it follows that none of $\{z_2d_1, d_1d_2, d_2d_3, d_3z_3\} $ belongs to~$M$. 
Hence $M \cap E(P) = \{e,e'\}$, and $|M \cap E(P)|$ is even. 
Now invoking \Cref{obs:sides} we conclude that $ \topclausevar{c}{j}$ and $u_{14}$ are on the same side of $M$.

\sloppy Now we prove $(b)$. Using \Cref{fig:clause-gadget} observe that there exists a path $P'$ between~$\botclausevar{a}{j}$ and $u_{8}$ such that $P'$ can be written as 
$\botclausevar{c}{j}~z_1~z_2~z_3~z_4~z_5~z_6~z_7~z_8~z_9~z_{10}$ $z_{11}~u_{8}$ where $\{z_1,z_2\}\subset V(F_1)$, 
$\{z_3,z_4,z_5\}\subset V(F_2)$, 
$\{z_6,z_7,z_8\}\subset V(F_3)$, and
$\{z_9,z_{10},z_{11}\}\subset V(F_4)$.
Now arguing similarly as in $(a)$ on the special 4-cycles $F_1,F_2,F_3,F_4$, we have that $|M \cap E(P)|$ is even. 
By \Cref{obs:sides}, we conclude that $ \botclausevar{a}{j}$ and $u_{8}$ are on the same side of $M$.
\end{proof}

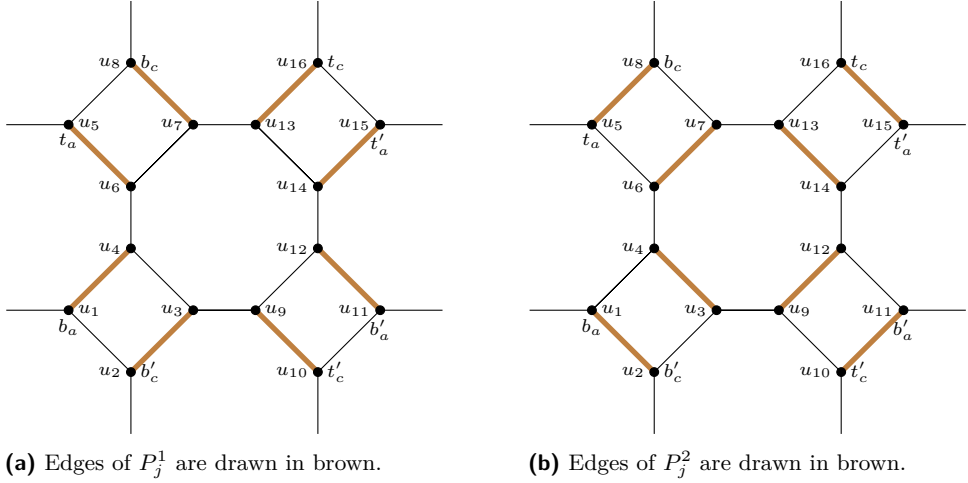
\begin{figure}[t]
  \centering
  \begin{subfigure}{0.45\textwidth}
    
\begin{tikzpicture}[scale=0.82]
      
\node[below] at (1,2)  {\scriptsize $b_{a}$};
\node[below] at (1,5)  {\scriptsize $t_{a}$};
\node[right] at (2,6)  {\scriptsize $b_{c}$};
\node[right] at (5,6)  {\scriptsize $t_{c}$};
\node[right] at (2,1)  {\scriptsize $b'_{c}$};
\node[right] at (5,1)  {\scriptsize $t'_{c}$};
\node[below] at (6,2)  {\scriptsize $b'_{a}$};
\node[below] at (6,5)  {\scriptsize $t'_{a}$};
\draw (1,2)--(2,1)--(3,2)--(2,3)--(1,2)--(2,3)--(2,4)--(1,5)--(2,6)--(3,5)--(2,4)--(3,5)--(4,5)--(5,4)--( 6,5)--(5,6)--(4,5)--(5,4)--(5,3)--(4,2)--(3,2)--(4,2)--(5,1)--(6,2)--(5,3);
\foreach \x/\y/\z/\t in {1/2/0/2, 1/5/0/5, 2/1/2/0, 2/6/2/7, 5/1/5/0, 5/6/5/7, 6/2/7/2, 6/5/7/5}
	\draw (\x, \y)--(\z, \t);
	
\foreach \x/\y/\z/\t in {1/2/2/3, 1/5/2/4, 2/1/3/2, 2/6/3/5, 5/1/4/2, 5/6/4/5, 6/2/5/3, 6/5/5/4}
	\draw[line width=0.07cm, brown] (\x, \y)--(\z, \t);
\foreach \x/\y [count = \n] in
 {1/2, 2/1, 3/2, 2/3, 1/5, 2/4, 3/5, 2/6, 4/2, 5/1, 6/2, 5/3, 4/5, 5/4, 6/5, 5/6}
	{ \filldraw (\x, \y) circle (2pt);
	\ifthenelse{\n = 5 \OR \n = 1 \OR \n = 9 \OR \n = 13}
	{\node[right] at (\x, \y) {\scriptsize $u_{\n}$};}
	{\node[left] at (\x, \y) {\scriptsize $u_{\n}$};}
	}

\end{tikzpicture}
\subcaption{Edges of $\Typecross{1}{j}$ are drawn in brown.}
\label{fig:cross-type-a}
\end{subfigure}
~~~
\begin{subfigure}{0.45\textwidth}
\begin{tikzpicture}[scale=0.82]

\draw (1,2)--(2,1)--(3,2)--(2,3)--(1,2)--(2,3)--(2,4)--(1,5)--(2,6)--(3,5)--(2,4)--(3,5)--(4,5)--(5,4)--( 6,5)--(5,6)--(4,5)--(5,4)--(5,3)--(4,2)--(3,2)--(4,2)--(5,1)--(6,2)--(5,3);
\node[below] at (1,2)  {\scriptsize $b_{a}$};
\node[below] at (1,5)  {\scriptsize $t_{a}$};
\node[right] at (2,6)  {\scriptsize $b_{c}$};
\node[right] at (5,6)  {\scriptsize $t_{c}$};
\node[right] at (2,1)  {\scriptsize $b'_{c}$};
\node[right] at (5,1)  {\scriptsize $t'_{c}$};
\node[below] at (6,2)  {\scriptsize $b'_{a}$};
\node[below] at (6,5)  {\scriptsize $t'_{a}$};
\foreach \x/\y/\z/\t in {1/2/0/2, 1/5/0/5, 2/1/2/0, 2/6/2/7, 5/1/5/0, 5/6/5/7, 6/2/7/2, 6/5/7/5}
	\draw (\x, \y)--(\z, \t);
	
\foreach \x/\y/\z/\t in {1/2/2/1, 1/5/2/6, 2/3/3/2, 2/4/3/5, 5/1/6/2, 5/6/6/5, 4/2/5/3, 4/5/5/4}
	\draw[line width=0.07cm, brown] (\x, \y)--(\z, \t);
\foreach \x/\y [count = \n] in
 {1/2, 2/1, 3/2, 2/3, 1/5, 2/4, 3/5, 2/6, 4/2, 5/1, 6/2, 5/3, 4/5, 5/4, 6/5, 5/6}
	{ \filldraw (\x, \y) circle (2pt);
	\ifthenelse{\n = 5 \OR \n = 1 \OR \n = 9 \OR \n = 13}
	{\node[right] at (\x, \y) {\scriptsize $u_{\n}$};}
	{\node[left] at (\x, \y) {\scriptsize $u_{\n}$};}
	}
\end{tikzpicture}
\subcaption{Edges of $\Typecross{2}{j}$ are drawn in brown.}
\label{fig:cross-type-b}
\end{subfigure}
\caption{$\Typecross{1}{j}$ and $\Typecross{2}{j}$ are the only possible restrictions of $M$ to a crossing gadget.}
\label{fig:cross-type}
\end{figure}

For any crossing gadget $X_j$ as drawn in \Cref{fig:cross-type} we consider the two perfect matching cuts $\Typecross{1}{j}$ of \Cref{fig:cross-type-a} and $\Typecross{2}{j}$ of \Cref{fig:cross-type-b} on~$X_j$.

\begin{lemma}\label{lem:crossing-gadget-cst}
Let $X_j$ be a crossing gadget of $\construct$.
For any $M \in \{\Typecross{1}{j}, \Typecross{2}{j}\}$, $M$~is a perfect matching cut of $X_j$.
Vertices $t_a, b_a, t'_a, b'_a$ are always on the same side of~$M$, and $t_c, b_c, t'_c, b'_c$ are always on the same side of~$M$. Moreover, if $M = \Typecross{1}{j}$, $t_a$ and $t_c$ are on the same side of $M$ (in $X_j$), otherwise they are not.
\end{lemma}
\begin{proof}
We refer to \Cref{fig:cross-type} for the notations on $X_j$.
$M$ is a perfect matching cut of $X_j$ by \Cref{lem:facial-cycle}.
Let $C$ be the external facial cycle of $X_j$.
We conclude by~\Cref{obs:sides} on subpaths of $C$ starting at $t_a$ or $t_c$.
\end{proof}

\subsection{Existence of perfect matching cut implies satisfiability}

In this section, we show that if $\construct$ has a perfect matching cut then $I$ has a~satisfying assignment. 
Let $M$ be a perfect matching cut of $M$ and $(A, B)$ be the cut of $M$.
As we already observed, a~potential solution to $I$ can be seen as a partition $(\mathcal{V}_A, \mathcal{V}_B)$ of the variables.
We set $x_i$ in $\mathcal{V}_A$ if and only if $V(\cycle{i}{2}) \subset A$, and we show that $\mathcal{P} = (\mathcal{V}_A, \mathcal{V}_B)$ satisfies the $\MSAT$-instance~$I$. 




Assume for contradiction that there exists a clause $C_j$ such that all variables, $x_a, x_b,x_c$, of $C_j$ are on the same side of $\mathcal{P}$. Thus all the vertices in $\cup_{i \in \{a, b, c\}} V(\cycle{i}{2})$ are on the same side of $M$. Assume without loss of generality that this side is~$A$. 
Let $z_i$ be any vertex of $\cycle{i}{2}$.

Now fix an integer $i\in \{a,b,c\}$. Observe, using \Cref{fig:var-gadget}, that there exists a path $P$ between $z_i$ and $\topvarclause{i}{j}$ such that $|M  \cap E(P)|$ is even. 
Hence, due to \Cref{obs:sides}, we infer that $\topvarclause{i}{j}$ lies in $A$. Now due to \Cref{lem:variable-clause-propagation} we have that $\left\{\topvarclause{i}{j}, \topclausevar{i}{j}, \botvarclause{i}{j}, \botclausevar{i}{j}\right\}\subset A$.
The above discussion implies that $$\displaystyle\bigcup\limits_{i\in \{a,b,c\}} \left\{\topvarclause{i}{j}, \topclausevar{i}{j}, \botvarclause{i}{j}, \botclausevar{i}{j}\right\}\subset A.$$ 
Now invoking \Cref{lem:intra-clause-propagation}, we have that all three vertices in $\{u_1,u_8,u_{14}\}$ lie in~$A$.
But this contradicts \Cref{lem:corner-not-same-side}.
Hence we get the following.

\begin{lemma}\label{lem:cut-to-SAT}
If $\construct$ has a~perfect matching cut then $I$ is a~positive instance. 
\end{lemma}

\subsection{Satisfiability implies the existence of a perfect matching cut}
In this section, we show that given a~$\MSAT$-instance~$I$ and a partition $\mathcal{P} = (\mathcal{V}_A, \mathcal{V}_B)$ satisfying~$I$, we can construct a~perfect matching cut $M_{\mathcal{P}}$ of $\construct$, as follows:

\begin{itemize}
	\item for each variable gadget $\varGadget{i}$, $M_{\mathcal{P}} \cap V(\varGadget{i})$ is the matching imposed by~\Cref{lem:var-connector},
	\item for each crossing gadget $X_j$ incident to two edges $e, f$, we choose $\Typecross{1}{j}$ if $\var(e)$ and $\var(f)$ are on the same side of $\mathcal{P}$, and $\Typecross{2}{j}$ otherwise,
	\item for each clause gadget $\clauseGadget{j}$ over variables $a, b, c$, we choose the matching of \Cref{fig:clause-gadget-type-1-match} if $b$ is not on the same side of $\mathcal{P}$ as $a$ and $c$,
the matching of \Cref{fig:clause-gadget-type-2-match} if $c$ is not on the same side of $\mathcal{P}$ as $a$ and $b$, and the matching of \Cref{fig:clause-gadget-type-3-match} in the last case.
\end{itemize}

As $M_{\mathcal{P}}$ is a perfect matching on each gadget, and as every vertex belongs to some gadget, $M_{\mathcal{P}}$ is a perfect matching of $\construct$.
By construction, $M_{\mathcal{P}}$ contains no connector edges.
Recall that any edge that does not have both endpoints inside the same gadget is a connector edge, we call \emph{connector vertex} a~vertex $v$ incident to a connector edge $e$, and that $\var(v) = \var(e)$.
\begin{lemma}\label{lem:side-change}
\sloppy For any path $Q$ between two connector vertices $u$ and $v$, we have $|Q \cap M_{\mathcal{P}}|$ even if and only if $\var(u)$ and $\var(v)$ are on the same side of $\mathcal{P}$.
\end{lemma}
\begin{proof}
As $M_{\mathcal{P}}$ does not contain any connector edges, $|Q \cap M_{\mathcal{P}}|$ is defined by the parts of $Q$ inside a gadget.
Let $Q_1, \dots, Q_k$ be spanning vertex-disjoint subpaths of $Q$ such that for any $i$, $Q_i$ lies inside a gadget and, there is a connector edge from the last vertex of $Q_i$ to the first one of $Q_{i+1}$, for every $1 \leq i < k$.
We prove the property by induction on $k$.

If $k = 1$, the whole $Q$ lies inside a gadget, and the property is true by~\Cref{lem:var-connector} for variable gadgets, \Cref{lem:corner-not-same-side} for clause gadgets and \Cref{lem:crossing-gadget-cst} for crossing gadgets.
 
Assume the property true for $i \leqslant k-1$, let $u'$ the last vertex of $Q_{k-1}$ and $v'$~the first vertex of $Q_k$.
By induction, $\var(u)$ and $\var(u')$ are on the same side of $\mathcal{P}$ if and only if $|\bigcup_{1 \leq i \leq k-1} E(Q_i) \cap M_{\mathcal{P}}|$ is even,
and $\var(v')$ and $\var(v)$ are on the same side of~$\mathcal{P}$ if and only if $|E(Q_k) \cap M_{\mathcal{P}}|$ is even.
As $\var(u') = \var(v')$, we know that $\var(u')$ is on the same side of $\mathcal{P}$ as $\var(v')$, moreover $u'v' \not \in M_{\mathcal{P}}$.
Thus $\var(u)$ and $\var(v)$ are on the same side of $\mathcal{P}$ if and only if $|\bigcup_{1 \leq i \leq k} E(Q_i) \cap M_{\mathcal{P}}|$ and $|E(Q_k) \cap M_\mathcal{P}|$ have the same parity, thus $|\bigcup_{1 \leq i \leq k} E(Q_i) \cap M_{\mathcal{P}}|$ is even if and only if $\var(u)$ and $\var(v)$ are on the same side of $\mathcal{P}$.
\end{proof}

\begin{lemma}\label{lem:SAT-PMC}
$M_{\mathcal{P}}$ is a perfect matching cut of $\construct$.
\end{lemma}

\begin{proof}
  We already know that $M_{\mathcal{P}}$ is a perfect matching.
  Moreover, $M_{\mathcal{P}}$ is a cutset by \Cref{lem:facial-cycle}.
Indeed, let $C$ be any cycle in $\construct$.
If $C$ is contained in a gadget then, as $M_{\mathcal{P}}$ is a cutset when restricted to a gadget, $|C \cap M_{\mathcal{P}}|$ is even.
Otherwise, $C$ contains a connector edge $uv$, so we can see $C$ as the concatenation of the edge $uv$ and a path $Q$ from $v$ to $u$.
We know that $uv \not \in M_{\mathcal P}$, and $\var(u) = \var(v)$.
By~\Cref{lem:side-change}, $|E(C) \cap M_{\mathcal{P}}| = |E(Q) \cap M_{\mathcal{P}}|$ is even.
\end{proof}

We finally get~\Cref{thm:hard}, due to \Cref{lem:cut-to-SAT,lem:SAT-PMC,lem:Barnette}.


\end{document}